\documentclass[11pt]{article}
\usepackage[T1]{fontenc}
\usepackage[utf8]{inputenc}
\usepackage{graphicx}
\usepackage[english]{babel}
\usepackage{amsmath}
\usepackage{amssymb}
\usepackage{amsfonts}
\usepackage{color}
\usepackage{babel}
\usepackage{verbatim}
\usepackage[unicode=true,pdfusetitle,bookmarks=true,bookmarksnumbered=false,bookmarksopen=false,
 breaklinks=false,pdfborder={0 0 1},backref=false,colorlinks=true]{hyperref}
\hypersetup{linkcolor=blue,citecolor=blue}
\makeatletter
\usepackage{pdflscape}
\usepackage{cite}
\usepackage{amsthm}

\newtheorem{FLRWCompatibleProperty}{Property}

\newtheorem{Lemma}{Lemma}
\newtheorem{definition}{Definition}
\newtheorem*{schurs_lemma}{Schur's Lemma}
\newtheorem{corollary}{Corollary}
\newtheorem{theorem}{Theorem}

\DeclareMathOperator{\Stab}{Stab}
\newcommand{\id}{\mathrm{id}}

\title{\bf FLRW-Compatible Tensors and Perfect-Fluid Universality in Modified Gravity}

\author{Cem Dalk\i rmaz$^{a,b}$\thanks{
email: cdalkirmaz@student.ethz.ch}, and Yaghoub Heydarzade$^c$\thanks{
email: yheydarzade@bilkent.edu.tr}, \\
\\{\small $^a$Department of Physics, ETH Zurich, 8093 Zurich, Switzerland}
\\{\small $^b$Department of Physics, Bilkent University, 06800 Ankara, T{\" u}rkiye}
\\{\small $^c$Department of Mathematics, Bilkent University, 06800 Ankara, T{\" u}rkiye}}
\date{}

\begin{document}

\maketitle

\begin{abstract}
We study FLRW-compatible tensor fields on Friedmann--Lema\^itre--Robertson--Walker spacetimes in $(D+1)$ dimensions equipped with Levi-Civita connection. We prove that any covariant rank-$N$ tensor field invariant under the spatial homogeneity and isotropy group, with $D>N$, admits a finite decomposition into monomials constructed from the metric tensor $g_{\mu\nu}$ and the unit timelike vector $u_\mu$, with scalar coefficients depending only on cosmic time. As a direct consequence, any rank-two correction term appearing in an Einstein-like rearrangement of the FLRW background field equations has the effective perfect-fluid form for $D>2$. Thus, the result provides an arbitrary-rank tensorial classification underlying the perfect-fluid structure of modified gravity equations on FLRW backgrounds, complementing existing rank-two universality results. Applications to generalized Friedmann equations and to the geodesic deviation equation in modified gravity are also discussed.
\end{abstract}

\section{Introduction}
\par To study the theory of general relativity, or a modified theory of gravity, in compliance with FLRW metric, it is standard practice to impose the homogeneity and isotropy to all dynamical fields of the theory. These are the tensor fields of any rank that are varied in the action, each having their corresponding equations of motion. For example, the metric tensor in usual General Relativity, the scalar field $\phi(t)$ in Brans-Dicke \cite{brans-dicke}, the aether field in Einstein-Aether Theory \cite{Einstein-Aether} are dynamical fields of their respective theories.
\par Consequently, when these diverse degrees of freedom are evaluated on FLRW spacetimes, powerful background level universality results emerge. In particular, it was shown in \cite{SVT-FLRW} that, once scalar and vector fields are restricted by homogeneity and isotropy, the metric field equations of a broad class of scalar--vector--tensor theories reduce to an Einstein-like form with an effective perfect-fluid source. On the other hand, when the theory lacks any dynamical field in the gravitational sector besides the metric tensor, the same result can be attained without such restriction. That is, for curvature driven theories such as $f(R)$, $f(R,\mathcal{G})$, Lovelock Gravity and any higher order gravity theories, the perfect fluid form of the field equations emerges dynamically from the properties of "perfect scalars" not only in FLRW \cite{metin}, but in Generalized Robertson-Walker (GRW) spacetimes with harmonic (divergence free) Weyl tensor \cite{PF in f(R), PF in higher order GR, PF in extended gravity}. Because the metric tensor is the sole fundamental field in these modified theories, the perfect fluid structure is solely dictated by the nature of the background geometry without imposing additional constraints on the dynamical fields of the matter sector.

\par To accommodate generalized theories containing arbitrary auxiliary fields (e.g. Einstein-Aether theory), our work approaches the problem from a complementary, symmetry-constrained tensorial perspective. Because these higher rank tensors and vectors do not automatically inherit the background symmetry, we impose what we call FLRW compatibility on all dynamical fields as an independent postulate, i.e. applying the Cosmological Principle to the fields themselves. In doing so, rather than restricting attention to scalar and vector degrees of freedom, symmetric rank-two field equations, or pure curvature objects, the group theoretic framework provides a comprehensive classification of rank-$N$ FLRW compatible tensor fields. By enforcing this algebraic constraint on the dynamical fields, any scalar quantity automatically qualifies as a "perfect scalar" according to the nomenclature introduced in \cite{PF in higher order GR}. The perfect fluid structure of rank-two correction terms in the field equations then follows as a corollary of this general monomial decomposition in a top-down manner.

\par In Section \ref{Section: FLRW Compatibility}, we begin by defining FLRW-compatible tensor fields and analyzing the form of dynamical scalar, vector, and rank-two tensor fields constrained by the Cosmological Principle in an FLRW geometry. Then, by means of a strong inductive argument, we prove Theorem \ref{theorem 1}. Informally, this theorem states that in a $(D+1)$-dimensional FLRW spacetime equipped with the Levi-Civita connection, any rank-$N$ tensor respecting homogeneity and isotropy (where $D>N$) is given by a linear combination of all possible rank-$N$ tensor products of the metric tensor $g_{\alpha\beta}$ and the timelike foliation vector $u_\alpha = -\delta_\alpha^0$. That is, any such tensor can be expressed as:

 \begin{equation}
 T_{\mu_1 \ldots \mu_N} =\sum\limits_{r}\sum\limits_{\sigma \in S_N}\alpha_{\sigma r}(t)\hspace{0.1cm}\sigma \cdot \left( g_{\mu_1\mu_2} \otimes \ldots \otimes g_{\mu_{2r-1}\mu_{2r}} \otimes u_{\mu_{2r+1}}\otimes \ldots \otimes u_N \right),
\end{equation}
 for some $\alpha_{\sigma r}(t) \in \mathbb{F}$, the underlying field of the vector space ($\mathbb{C}$ or $\mathbb{R}$), and the action of $\sigma \in S_N$ (the symmetric group) permutes the order of indices in the product, yielding all possible combinations.
 
After reviewing FLRW geometry in section \ref{section: review of FLRW geometry}, we analyze the field equations of general modified theories of gravity obtained by the variation of actions of the following form in Section \ref{Section: field equations of generalized GR in cosmology}.
 \begin{equation}
 S = \int d^{(D+1)}x\sqrt{-g}\left\{ \mathcal{L}_{fields}\left(g_{\mu\nu}, R^\rho_{\sigma\mu\nu}, \nabla_\alpha, A^{(i)}_{\mu_1\ldots\mu_k}, \ldots\right)+\mathcal{L}_{matter}\right\}, \label{generalized action}
\end{equation}
 where $\mathcal{L}_{matter}$ is the usual matter Lagrangian and $\mathcal{L}_{fields}$ is the generalized version of Einstein--Hilbert Lagrangian, a scalar function constructed from the contractions of metric tensor $g_{\mu\nu}$, Riemann tensor $R^\rho_{\sigma\mu\nu}$, rank-k dynamical tensor fields $A^{(i)}_{\mu_1\ldots\mu_k}$ labeled by $i$ with $k \in \{0,\ldots, N\}$ (these are specific to the modified theory), their covariant derivatives of any order and their tensor products. Then, using Theorem \ref{theorem 1} and Theorem 5 in \cite{metin} (which states the form of the field equations of generic modified theories of gravity in FLRW geometry when the action is constructed exclusively from the Riemann tensor and its covariant derivatives), we state the following tensorial consequence.

\noindent
\begin{corollary} \label{main corollary} Suppose that the metric field equation obtained by varying the action in \eqref{generalized action} can be rearranged into an Einstein-like form \begin{equation} G_{\mu\nu} = \frac{1}{\xi}T_{\mu\nu} + Q_{\mu\nu},\end{equation} where $G_{\mu\nu}$ and $T_{\mu\nu}$ are the Einstein tensor and the usual energy--momentum tensor, respectively; $\xi$ is a scalar function constructed from contractions of the dynamical tensor fields appearing in the theory, their covariant derivatives of any order, and their tensor products; and $Q_{\mu\nu}$ is the remaining second-rank correction term constructed from similar contractions and products. Then, in a $(D+1)$-dimensional FLRW spacetime with $D>2$, $Q_{\mu\nu}$ necessarily has the perfect-fluid form \begin{equation} Q_{\mu\nu} = \alpha(t)g_{\mu\nu} + \beta(t)u_\mu u_\nu, \end{equation} for some time-dependent scalar functions $\alpha(t)$ and $\beta(t)$. Moreover, $\xi$ is spatially homogeneous and therefore depends only on cosmic time, \begin{equation} \xi = \xi(t). \end{equation} \end{corollary}

\par Corollary \ref{main corollary} is a background-level tensorial statement. It does not imply that all modified gravity theories are dynamically equivalent on FLRW backgrounds. Rather, it states that once the equations are restricted to an FLRW-compatible sector, the rank-two tensorial structure of the metric field equation is forced into the perfect-fluid sector. The theory dependence remains in the scalar functions multiplying $g_{\mu\nu}$ and $u_\mu u_\nu$, and in the additional field equations of the particular theory.

\par To appreciate the scope of this background-level restriction, it is useful to contrast Corollary \ref{main corollary} with Lovelock’s theorem \cite{lovelock}, which classifies the symmetric rank-two tensors built locally
from the metric and its derivatives up to second order with identically vanishing
divergence: these are the Lovelock tensors in arbitrary dimension, and
$a\,G_{\mu\nu}+b\,g_{\mu\nu}$ with constant $a,b$ in four dimensions. A
perfect-fluid source is not an assumption of that theorem; it follows once the
field equations are specialized to a homogeneous and isotropic background, since
the left-hand side is then built solely from $g_{\mu\nu}$ and $u_\mu$.

\par Our classification answers a complementary question. Fixing the geometry to be FLRW, we classify the tensor fields compatible with its
symmetry; these need not be built from the metric alone, need not be of bounded
derivative order, and need not be divergence-free, the divergence-free condition
playing no role in the proof of Theorem \ref{theorem 1}. Accordingly, the
coefficients of Lovelock's classification are constants, while ours are functions
of cosmic time. The two intersect on FLRW backgrounds in both directions. Every
Lovelock tensor is built from the metric alone, hence is FLRW-compatible, and by
Corollary \ref{main corollary} reduces to
$\alpha(t)g_{\mu\nu}+\beta(t)u_\mu u_\nu$ for any $D>2$, recovering the explicit
general Lovelock and $\mathcal{F}(R,\mathcal{G})$ cases of \cite{metin} as
instances of the general classification (also see \cite{PF in higher order GR}). Conversely, since Lovelock's theorem
implies that any genuine modification of general relativity in four dimensions
must introduce additional dynamical fields, Theorem \ref{theorem 1} shows that
such fields, whatever their rank, contribute no new tensorial structure on an
FLRW background, the theory dependence surviving entirely in the scalar functions
multiplying $g_{\mu\nu}$ and $u_\mu u_\nu$. 
\par While our framework fixes the geometry to deduce the tensorial form, a closely related but logically reversed question is addressed in \cite{PF to be GRW}, where a condition is given for a perfect-fluid spacetime to
be a generalized Robertson--Walker spacetime. The logical direction there is the
reverse of the one taken here: the perfect-fluid form of the Ricci tensor is the
hypothesis and the geometry is the conclusion, whereas in Corollary
\ref{main corollary} the geometry is the hypothesis and the tensorial form is the
conclusion. The two settings differ also in strength. A generalized
Robertson--Walker spacetime is characterized by a timelike torse-forming unit
vector rather than by the full spatial isometry group, so the results of
\cite{PF to be GRW} rest on weaker geometric assumptions and do not require the
fields to inherit the background symmetry, while Theorem \ref{theorem 1} holds
for tensors of arbitrary rank and for arbitrary dynamical field content but is
confined to the exactly symmetric background. Extending the present
classification to generalized Robertson--Walker geometries, where the spatial
isometry group underlying our proof is no longer available, is therefore a
natural direction that the symmetry-based method used here does not cover.

\par In the last section \ref{Section GDE in modified GR}, we utilize Corollary \ref{main corollary} to analyze the geodesic deviation equation in modified theories of gravity for cosmological solutions.

\par Throughout the paper, we will be working in a $D+1$ spacetime with mostly-plus signature, and denote the time coordinate $t$ by $x^0$, and the spatial coordinates by Latin indices, i.e. $x^i$ with $i = 1,..., D$.

\section{FLRW-Compatible Tensor Fields} \label{Section: FLRW Compatibility}
\par We start with the following definition.

\textbf{Definition:} A tensor field on an FLRW spacetime is called
\textit{FLRW-compatible} if it is invariant under the full isometry group of the
FLRW background, equivalently, if its Lie derivative vanishes along every Killing
vector field generating the spatial homogeneity and isotropy of the FLRW geometry.

 For an $N$ rank FLRW compatible tensor field $T_{\mu_1 \ldots \mu_N}$, this implies
\begin{align}
 \mathcal{L}_\mathcal{K}\left(T_{\mu_1 \ldots \mu_N}\right) = 0,
\end{align}
for any spatial Killing vector $\mathcal{K}$ of the FLRW metric. In other words, the Lie derivative of an FLRW compatible tensor vanishes along any arbitrary spatial Killing vector, i.e., an infinitesimal generator $\mathcal{K}^\mu$ of an isometry of the FLRW metric with $\mathcal{K}^0 = 0$. In this regard, all dynamical fields of the theory are restricted to the FLRW compatibility condition.
\par With these definitions, some of the properties of FLRW compatible tensors are as follows:

\begin{FLRWCompatibleProperty}\label{FLRW Compatible statement 1}
The contraction of two FLRW-compatible tensors is FLRW-compatible.
\end{FLRWCompatibleProperty}
\begin{proof}
 Given two FLRW compatible tensors, $A_{\mu_1 \ldots \mu_N}$ and $B_{\nu_1 \ldots \nu_M}$, 
 \begin{align}
 \mathcal{L}_\mathcal{K}\left( A_{\mu_1 \ldots \mu_N}B_{\nu_1 \ldots \nu_M}g^{\mu_i\nu_j} \right) = &\mathcal{L}_\mathcal{K}\left( A_{\mu_1 \ldots \mu_N} \right)B_{\nu_1 \ldots \nu_M}g^{\mu_i\nu_j} \notag\\&+A_{\mu_1 \ldots \mu_N} \mathcal{L}_\mathcal{K}\left( B_{\nu_1 \ldots \nu_M}\right )g^{\mu_i\nu_j} + A_{\mu_1 \ldots \mu_N} B_{\nu_1 \ldots \nu_M}\ \mathcal{L}_\mathcal{K}\left(g^{\mu_i\nu_j}\right) \notag \\ &= 0.
\end{align}
 Because, by definition, $\mathcal{L}_\mathcal{K}\left( g_{\alpha\beta} \right) = 0$ and $\mathcal{L}_\mathcal{K}\left( g^{\alpha\beta} \right) = 0$.
\end{proof}

\begin{FLRWCompatibleProperty}\label{FLRW Compatible statement 2}
The tensor product of two FLRW-compatible tensors is FLRW-compatible.
\end{FLRWCompatibleProperty}
\begin{proof}
 For two such tensors as before,
 \begin{align}
 \mathcal{L}_\mathcal{K}\left( A_{\mu_1 \ldots \mu_N} \otimes B_{\nu_1 \ldots \nu_M} \right) = \mathcal{L}_\mathcal{K}\left( A_{\mu_1 \ldots \mu_N}\right)\otimes B_{\nu_1 \ldots \nu_M} + A_{\mu_1 \ldots \mu_N}\otimes\mathcal{L}_\mathcal{K}\left(B_{\nu_1 \ldots \nu_M}\right) = 0.
\end{align}
\end{proof}
\begin{FLRWCompatibleProperty}\label{FLRW Compatible statement 3}
The covariant derivative of an FLRW-compatible tensor is FLRW-compatible.
\end{FLRWCompatibleProperty}
\begin{proof}
 Assuming metric compatibility, we have
 \begin{equation}
 \nabla_\alpha g_{\mu\nu} = 0.
\end{equation}
 In addition,
 \begin{equation}
 \nabla_{\alpha}u_\beta = H(t)h_{\alpha\beta} = H(t)(g_{\alpha\beta} + u_{\alpha}u_\beta),
\end{equation}
 
 \noindent where $H(t)$ is the Hubble parameter; and
 \begin{equation}
 \nabla_\alpha \theta(t) = -u_\alpha\dot{\theta}(t),
\end{equation}
 for any scalar field $\theta(t)$ which depends only on time.
 Then, the rest follows from Theorem \ref{theorem 1}. (We can safely use Theorem \ref{theorem 1} here, since its proof is independent of this property.)
\end{proof}

\par We remark that the same statement can also be justified directly from the invariance of the Levi--Civita connection under isometries, without invoking Theorem \ref{theorem 1}. If $\mathcal{K}$ is a Killing vector of the FLRW background, then the flow generated by $\mathcal{K}$ preserves $g_{\mu\nu}$ and hence also preserves its Levi--Civita connection, see Chapter 1 in \cite{kentaro}. Equivalently,
\begin{equation}
    \mathcal{L}_{\mathcal{K}}\Gamma^{\rho}_{\alpha\beta}=0.
\end{equation}
For a covariant tensor $T_{\mu_1\cdots\mu_N}$ one therefore has
\begin{align}
    \mathcal{L}_{\mathcal{K}}\left(\nabla_\alpha T_{\mu_1\cdots\mu_N}\right)=\nabla_\alpha\left(\mathcal{L}_{\mathcal{K}}T_{\mu_1\cdots\mu_N}\right)-\sum_{j=1}^N\left(\mathcal{L}_{\mathcal{K}}\Gamma^{\rho}_{\alpha\mu_j}\right)T_{\mu_1\cdots\rho\cdots\mu_N}.
\end{align}
Both terms vanish whenever $T_{\mu_1\cdots\mu_N}$ is FLRW-compatible. Thus $\nabla_\alpha T_{\mu_1\cdots\mu_N}$ is FLRW-compatible. This provides an independent geometric proof of Property \ref{FLRW Compatible statement 3}; a complementary invariant-theoretic proof of the full monomial decomposition is given in Appendix \ref{appendix:invariant-proof}.

It is also worth to compare FLRW compatibility with the
algebraic notion of Weyl compatibility introduced in \cite{weyl-compatibility}, where a symmetric tensor $b_{\mu\nu}$ is called Weyl compatible if $b_{\mu\rho}C^{\rho}{}_{\nu\alpha\beta}$ has vanishing cyclic sum over its first three indices, and Riemann compatible if the same holds with Weyl replaced by the Riemann tensor \cite{riemann-compatibility}. Although Weyl and Riemann compatibility measures are well suited for characterizing curvature based constraints on tensor fields, they prove inadequate for a symmetry based analysis reflecting the full symmetry group of the background geometry. Specifically, Weyl compatibility becomes trivial in FLRW geometry as the Weyl tensor vanishes, while Riemann compatibility in FLRW only dictates commutation with the Ricci tensor \cite{weyl-compatibility}. This only imposes the foliation vector $u_{\alpha}$ to be an eigenvector of the given tensor $A_{\alpha\beta}$, i.e  $u_\beta A_{\alpha\gamma} u^\gamma = u_\alpha A_{\beta\gamma} u^\gamma$, which successfully eliminates space-time cross components of the tensor ($A_{0j}$), but leaves the spatial sector unconstrained. Therefore, Riemann compatibility is mathematically insufficient to recover the isotropic perfect-fluid structure we will get from imposing FLRW compatibility. Furthermore, it lacks a natural generalization beyond rank-two tensors, as it naturally follows from the definition of a Codazzi tensor. In contrast, the Lie-derivative approach under FLRW compatibility provides a direct, arbitrary-rank classification tailored to reflect the background symmetry.

\par In the following sections, we obtain the explicit forms of FLRW-compatible scalar fields, vector fields, and rank-two tensor fields, building on Chapter~4 in \cite{homogeneity}. The scalar and vector field results
will serve as the base steps for the inductive argument to prove Theorem \ref{theorem 1}, and the second rank case result will be a proof of concept for higher orders.

\subsection{Scalar Fields}
\par We consider an FLRW-compatible scalar field $\phi(x)$. To impose homogeneity and isotropy, we demand that the field is Lie transported along an arbitrary spatial isometry of the metric. In other words, Lie derivative of the scalar field $\phi$ must vanish along an arbitrary spatial Killing vector $\mathcal{K}$, i.e., $\mathcal{K}^0 = 0$. 
\par Taking $\mathcal{K}(x)$ to be such a spatial Killing vector, the killing-type equation for the field $\phi$ is,
\begin{align}
 \mathcal{L}_\mathcal{K}\phi(x) = \mathcal{K}^\alpha \partial_\alpha \phi(x) = \mathcal{K}^i\partial_i \phi(x) = 0, \label{killing equation for a scalar field}
\end{align}
where the Greek indices such as $\alpha$ run over $0,.., D$, i.e., all the components, and the Latin indices such as $i$ run over $1,..., D$ spatial components only.
\par Because the Killing vector $\mathcal{K}^i(x)$ is arbitrary, 
\eqref{killing equation for a scalar field} implies:
\begin{align}
 \partial_i\phi(x) = 0,
\end{align}
that is, $\phi(x) = \phi(t)$. Any dynamical scalar field respecting the cosmological principle has only time dependence.

\subsection{Vector Fields} \label{section: vector fields}
\par Considering an FLRW-compatible vector field $A^\alpha$, the same requirement applies. We demand that the vector field is Lie transported along an arbitrary spatial isometry, $\mathcal{K}(x)$ with $\mathcal{K}^0(x) = 0$, of the metric.

\begin{equation}
 \mathcal{L}_\mathcal{K}A^\alpha = \mathcal{K}^\beta(x)\partial_\beta A^\alpha - A^\beta\partial_\beta \mathcal{K}^\alpha(x) = 0.\label{killing equation for a vector field}
\end{equation}
\par Setting $\alpha = 0$, and using $\mathcal{K}^0 = 0$, the killing-type equation \eqref{killing equation for a vector field} for the vector field $A^\alpha$ reduces to the following.
\begin{equation}
 \mathcal{K}^\beta(x)\partial_\beta A^0 = \mathcal{K}^i(x)\partial_i A^0 = 0. \label{time component of vector field killing equation}
\end{equation}
As the Killing vector $\mathcal{K}^i(x)$ is arbitrary, \eqref{time component of vector field killing equation} implies
\begin{equation}
 \partial_iA^0 = 0,
\end{equation}
where $x^i$ represent the spatial coordinates, hence 
\begin{equation}
 A^0 = A^0(t),
\end{equation}

\par As the next step, we set $\alpha = i$, i.e., the spatial coordinates.
\begin{equation}
 \mathcal{K}^j(x)\partial_j A^i - A^\beta\partial_\beta \mathcal{K}^i(x) = 0. \label{vector field killin eqn}
\end{equation}
\par Because the FLRW geometry is maximally symmetric and isotropic about every point, at an arbitrary point P there always exists a rotational Killing vector centered at P and vanishing at the same point P. Furthermore, its covariant derivative evaluated at the same arbitrary point P is also arbitrary, only subject to antisymmetry by the killing equation (see p. 378 in \cite{weinberg}). Denoting such arbitrarily chosen rotational Killing vector on the effectively three-dimensional spatial hypersurface and centered at P, by $\mathcal{R}^m$, this statement implies,
\begin{align}
 \left.\mathcal{R}^m\right|_P = 0, \label{vanishing of rotational Killing vector at a point}
\end{align}
and $\left (\nabla_\alpha \mathcal{R}^m)\right|_{P} = \left (\partial_\alpha \mathcal{R}^m)\right|_{P}$ being arbitrary at point P.

Then, writing the equation \eqref{vector field killin eqn} for a rotational Killing vector $\mathcal{R}^m$ centered at an arbitrary point P, and evaluating it at the same point gives
\begin{equation}
 \left.(A^\beta\partial_\beta\mathcal{R}^i)\right|_P = 0, \label{Vector field killing eqn step 1}
\end{equation}
where the first term vanishes by \eqref{vanishing of rotational Killing vector at a point}. \par Furthermore, because we are working locally at some point P, the killing equation becomes
\begin{equation}
 \left(\partial_\beta\mathcal{K}^\alpha + \partial^\alpha \mathcal{K}_\beta\right)|_P = 0. \label{local killing equation}
\end{equation}
\par In addition, $\mathcal{K}_0 = g_{0\alpha} \mathcal{K}^\alpha = g_{00} \mathcal{K}^0 = 0$ for a diagonal metric, as is the case in FLRW geometry. Thus, by \eqref{local killing equation}, at a local point P we have the following
\begin{align}
 \partial_0 \mathcal{K}^\alpha |_P = -\partial^\alpha \mathcal{K}_0 |_P= 0.
\end{align}
As a result, $A^0\partial_0 \mathcal{R}^i$ term in \eqref{Vector field killing eqn step 1} gives equivalently zero, yielding:
\begin{align}
 \left(A^j\partial_j\mathcal{R}^i\right)|_P= 0. \label{Vector field killing eqn step 2}
\end{align}

In the remainder of the calculation, we will omit explicitly indicating evaluation at the arbitrary point P, as it is understood from the context.
Using \eqref{local killing equation} in \eqref{Vector field killing eqn step 2} we manipulate the equation as follows,
\begin{align}
 A^j\partial^i \mathcal{R}_j = 0, \notag\\
 A^j\delta^i_n\partial^n \mathcal{R}_j = 0,\notag\\
 (A^j\delta^i_n)\frac{1}{2}\left( \partial^n\mathcal{R}_j - \partial_j\mathcal{R}^n\right) = 0,\notag\\
 (A^j\delta^i_n)\partial^n\mathcal{R}_j - (A^j\delta^i_n)\partial_j\mathcal{R}^n = 0,
\end{align}
where we again used the antisymmetry of $\partial_i \mathcal{R}^j$ by \eqref{local killing equation}. Now, interchanging the dummy indices $j$ and $n$ in the second term and combining gives, with $h_{\alpha\beta}$ the induced metric on the spatial hypersurface,
\begin{equation}
 \partial^n\mathcal{R}_j\left( A^j\delta^i_n - A_nh^{ij} \right) = 0. \label{Vector field homogeneity equation 1}
\end{equation}
Since $\partial^n\mathcal{R}_j$ is arbitrary and the term in the brackets is antisymmetric, equation \eqref{Vector field homogeneity equation 1} implies
\begin{equation}
 A^j\delta^i_n - A_nh^{ij} = 0.
\end{equation}
Contracting $i$ and $n$, we get:
\begin{equation}
 (D-1)A^j = 0.
\end{equation}
Thus, unless we have only a single spatial dimension in our theory, $A^j = 0$, that is the spatial part of our vector field must vanish.
\par As a result, the dynamical vector field must have the form
\begin{equation}
 A^\alpha = A^0(t)u^\alpha,
\end{equation}
in other words, parallel to the foliation vector field $u^\alpha = \delta^\alpha_0$. Normally, this result is only valid at the specified arbitrary point P. However, owing to homogeneity, it is valid throughout the entire spacetime.
\par Therefore, any dynamical vector field $A^\alpha(x)$ in the theory, complying with the cosmological principle, has the following form:
\begin{equation}
 A^\alpha(x) = \lambda(t)u^\alpha, \label{general vector field result}
\end{equation}
where $\lambda(t)$ is a scalar function dependent on the time coordinate only.
\subsection{Second Rank Tensor Fields} \label{section: second-rank tensor fields}

\par We use a similar argument for a second rank FLRW compatible tensor field $A_{\alpha\beta}$. As before, we impose that the tensor field is Lie transported along an arbitrary spatial isometry.
\begin{align}
 \mathcal{L}_\mathcal{K}A_{\alpha\beta} =\mathcal{K}^\gamma\partial_\gamma A_{\alpha\beta} + A_{\gamma\beta}\partial_\alpha \mathcal{K}^\gamma + A_{\alpha \gamma}\partial_\beta \mathcal{K}^\gamma = 0. \label{killing equation for second-rank tensor}
\end{align}
\par We consider $\alpha =i$, $\beta = j$, spatial coordinates. Since $\mathcal{K}$ is a spatial Killing vector, $\gamma = k$, spatial components.
\begin{equation}
 \mathcal{K}^k\partial_k A_{ij} + A_{k j}\partial_i \mathcal{K}^k + A_{i k}\partial_j \mathcal{K}^k = 0, \label{2nd rank tensor killing eqn}
\end{equation}

Again, we consider an arbitrary rotational spatial Killing vector $\mathcal{R}^k$ centered at P, thus $\left(\mathcal{R}^k\right)|_P = 0$ with $\left(\partial_j \mathcal{R}^k)\right|_{P}$ arbitrary at an arbitrary point P.
Then, evaluating the equation \eqref{2nd rank tensor killing eqn} at P gives
\begin{align}
 A_{kj}\partial_i \mathcal{R}^k + A_{i k}\partial_j \mathcal{R}^k &= 0, \notag\\
 (A_{kj}\delta^n_i)\partial_n\mathcal{R}^k + (A_{ik}\delta_j^n)\partial_n \mathcal{R}^k &= 0, \notag\\
 (A_{kj}\delta^n_i)\frac{1}{2}\left( \partial_n\mathcal{R}^k - \partial^k\mathcal{R}_n \right) + (A_{i k}\delta_j^n)\frac{1}{2}\left( \partial_n\mathcal{R}^k - \partial^k\mathcal{R}_n \right) &= 0,\notag\\
 (A_{kj}\delta^n_i)\partial_n\mathcal{R}^k - (A_{kj}\delta^n_i)\partial^k\mathcal{R}_n + (A_{ik}\delta_j^n)\partial_n\mathcal{R}^k - (A_{ik}\delta_j^n)\partial^k\mathcal{R}_n&=0,
\end{align}
where we used the antisymmetry of $\partial_m\mathcal{R}^i$ by the local killing equation \eqref{local killing equation}. Now, interchanging the dummy indices $n$ and $k$ in the second and the fourth term, and combining the terms yield:
\begin{equation}
 \left( A_{kj}\delta^n_i - A^n_{j}h_{ki}\right)\partial_n\mathcal{R}^k + \left( A_{ik }\delta^n_j -A^n_i h_{kj}\right)\partial_n\mathcal{R}^k = 0. \label{tensor field homogeneity equation 1}
\end{equation}
Because $\partial_n\mathcal{R}^k$ is arbitrary, equation \eqref{tensor field homogeneity equation 1} implies:
\begin{equation}
 A_{kj}\delta^n_i - A^n_{j}h_{ki} + A_{ik }\delta^n_j -A^n_i h_{kj} = 0.
\end{equation}
Contracting $i$ and $n$, we get
\begin{align}
 A_{kj}D - A_{kj} + A_{jk} - A^{(3)}h_{kj} &= 0, \notag\\
 A_{kj}(D-1) + A_{jk} &= A^{(3)}h_{kj}, \label{second-rank tensor, dimension eqn}
\end{align}
where $A^{(3)}$ is the trace of $A$ over the spatial part. Interchanging the indices $k$ and $j$ and removing from \eqref{second-rank tensor, dimension eqn} gives,
\begin{equation}
 (D-2)(A_{kj} - A_{jk}) = 0.
\end{equation}
Thus, unless we have only 2 spatial dimensions in our theory, spatial part of any second-rank tensor subject to cosmological principle must be symmetric.
\par Using this result in \eqref{second-rank tensor, dimension eqn}, $A_{kj}$ can be written as
\begin{equation}
 A_{kj} = \frac{A^{(3)}}{D}h_{kj},
\end{equation}
where $h_{kj}$ is the induced spatial metric. As before, this result is only valid at an arbitrary local point P, however, homogeneity requires it to hold throughout the manifold. Therefore, it applies at any arbitrary point $x^\mu$ as well.
\par To check the coordinate dependence of the trace $A^{(3)}$, we note that it is a dynamical scalar field as well by the product rule for the Lie derivative.
\begin{align}
 \mathcal{L}_\mathcal{K}\left( A^{(3)}\right) = \mathcal{L}_\mathcal{K}\left( h^{ij}A_{ij}\right) = \mathcal{L}_\mathcal{K}(h^{ij})A_{ij} + \mathcal{L}_\mathcal{K}(A_{ij})h^{ij} = 0,
\end{align}
where the terms vanish because $h_{ij}$ and $A_{ij}$ are FLRW-compatible fields.
Therefore,
\begin{equation}
 A^{(3)} = A^{(3)}(t).
\end{equation}
In the end, any dynamical second-rank tensor field has the following form:
\begin{equation}
[A_{\alpha\beta}] = 
\left(
\begin{array}{c|c c c}
 a & b_1 & b_2 & b_3 \\
 \hline
 c_1& & & \\
 c_2 & & 
 \frac{A^{(3)}(t)}{D}[h_{\alpha\beta}] \\
 c_3 & & & 
\end{array}
\right).
\end{equation}
To obtain the restriction on the terms residing in first column and first row, we can use our result for a vector field, given in \eqref{general vector field result}. By using the foliation vector field $u^\alpha =-\delta^\alpha_0$ respecting the cosmological principle, we can obtain two different vector fields from $A_{\alpha\beta}$
\begin{align}
 v_\alpha &= A_{\alpha\beta}u^\beta = A_{\alpha 0}, \notag\\
 w_\beta &= A_{\alpha\beta}u^\alpha = A_{0\beta}.
\end{align}
Because of the product rule of Lie derivative,
\begin{align}
 \mathcal{L}_K(A_{\alpha\beta} u^\beta) &= \mathcal{L}_K(A_{\alpha\beta}) u^\beta + A_{\alpha\beta}\mathcal{L}_K({u^\beta}), \notag\\
 \mathcal{L}_K(v_\alpha) &= 0.
\end{align}
In other words, vector fields obtained by contracting dynamical fields that respect the cosmological principle also adhere to the cosmological principle. Thus, using \eqref{general vector field result}, $w_\beta = \gamma(t)u_\beta = -\gamma(t)\delta_\beta^0$ and $v_\alpha = \phi(t)u_\alpha = -\phi(t)\delta_\alpha^0$ for some scalar functions $\gamma(t)$ and $\phi(t)$. Therefore,
\begin{align}
 A_{\alpha0} &= -\phi(t)\delta^0_\alpha, \notag\\
 A_{0\beta} &= -\gamma(t)\delta^0_\beta.
\end{align}
We obtain, $A_{00} = a = -\phi(t) = -\gamma(t)$ and $A_{0i} = 0$, $A_{i0} = 0$.
\par Consequently, any dynamical second-rank tensor field $A_{\alpha\beta}(x)$ in the theory, respecting the cosmological principle, has the following form
\begin{align}
 A_{\alpha\beta}(x) &= A_{00}(t)u_{\alpha}u_\beta +\frac{A^{(3)}(t)}{D}h_{\alpha\beta}, \notag\\
 &=\left(A_{00}(t) + \frac{A^{(3)}(t)}{D}\right)u_\alpha u_\beta + \frac{A^{(3)}(t)}{D}g_{\alpha\beta},
\end{align}
where we used $h_{\alpha\beta} = g_{\alpha\beta} + u_\alpha u_\beta$, stemming from the foliation with $u^\alpha u_\alpha = -1$. Combining the terms gives
\begin{equation}
 A_{\alpha\beta}(x) = \alpha(t)u_\alpha u_\beta + \beta(t)g_{\alpha\beta}. \label{general second-rank tensor field result}
\end{equation}

\subsection{Higher Rank Tensor Fields}
\par Having obtained the preliminary results for rank-zero,rank-one and rank-two FLRW compatible tensor fields, we will utilize an inductive argument to obtain the form of higher rank FLRW compatible tensor fields.

\subsubsection{Symmetry Constraint on The Tensor}

\par Considering a rank-$N$ dynamical tensor field $T_{\mu_1...\mu_N}$, we impose it to be Lie transported along an arbitrary spatial isometry, that is
\begin{align}
 \mathcal{L}_\mathcal{K}T_{\mu_1...\mu_N} &= \mathcal{K}^\alpha(\partial_\alpha T_{\mu_1...\mu_N}) +T_{\alpha\mu_2...\mu_N}(\partial_{\mu_1}\mathcal{K}^\alpha) + T_{\mu_1\alpha\mu_3...\mu_N}(\partial_{\mu_2}\mathcal{K}^\alpha) + \ldots  + T_{\mu_1\mu_2...\mu_{N-1}\alpha}(\partial_{\mu_N}\mathcal{K}^\alpha), \notag\\[5pt]
 &= \mathcal{K}^\alpha(\partial_\alpha T_{\mu_1...\mu_N}) + \sum_{i=1}^N \left[T_{...\mu_{i-1}\alpha\mu_{i+1}...}(\partial_{\mu_i}\mathcal{K}^\alpha)\right] = 0. \label{killing equation for N rank tensor}
\end{align}
\par We begin with the fully spatial part of the tensor, i.e., $T_{a_1...a_N}$ with $a_i \in {1, \ldots, D}$. Taking the Killing vector $\mathcal{K}^k$ to be an arbitrary rotational Killing vector $\mathcal{R}^k$ at P, hence $(\mathcal{R}^k)|_P = 0$, $(\partial_j\mathcal{R}^k)|_P$ arbitrary at an arbitrary point P and $\partial_0\mathcal{R}^k = 0$ from our previous discussion, we evaluate \eqref{killing equation for N rank tensor} at P and obtain
\begin{align}
 \left.\sum_{i=1}^N \left[T_{...a_{i-1} m a_{i+1}...}(\partial_{a_i}\mathcal{R}^m)\right]\right|_P = \left.\sum_{i=1}^N \left[T_{...a_{i-1} m a_{i+1}...}\delta_{a_i}^n(\partial_{n}\mathcal{R}^m)\right]\right|_P = 0.
\end{align}
As before, in the remainder of the calculation, we omit indicating the evaluation at point P, as this is clear from the context.
Now, using the local killing equation \eqref{local killing equation}, we transfer the antisymmetry of $\partial_n\mathcal{R}^m$ to the tensor field. 

\begin{align}
 (\partial_n\mathcal{R}^m)\sum_{i=1}^N\left[ T_{\ldots a_{i-1} m a_{i+1}\ldots}(\delta_{a_i}^n) - T_{\ldots a_{i-1}}{}^n{}_{a_{i+1}\ldots}(h_{a_im})\right] = 0. \label{N rank tensr analysis step 1}
\end{align}
In \eqref{N rank tensr analysis step 1}, $\partial_n\mathcal{R}^m$ is arbitrary at P, yielding,
\begin{align}
 \sum_{i=1}^N\left[ T_{\ldots a_{i-1} m a_{i+1}\ldots}(\delta_{a_i}^n) - T_{\ldots a_{i-1}}{}^n{}_{a_{i+1}\ldots}(h_{a_im})\right] = 0. \label{N rank tensr analysis step 2}
\end{align}
We separate the $i=j$ for a fixed $j \in \{1, \ldots, N\}$:
\begin{align}
 &\sum_{\substack{i = 1 \\ i \ne j}}^N\left[ T_{\ldots a_{i-1} m a_{i+1}\ldots}(\delta_{a_i}^n) - T_{\ldots a_{i-1}}{}^n{}_{a_{i+1}\ldots}(h_{a_im})\right]+T_{\ldots a_{j-1}ma_{j+1}\ldots}(\delta^n_{a_j}) - T_{\ldots a_{j-1}}{}^n{}_{a_{j+1}\ldots}h_{a_jm} = 0.
\end{align}
Contracting the indices $n$ and $a_j$ gives,
\begin{align}
 &\sum_{\substack{i = 1 \\ i \ne j}}^N\left[ T_{\ldots a_{i-1} m a_{i+1}\ldots a_{j-1}a_ia_{j+1}\ldots} - T_{\ldots a_{i-1}}{}^n{}_{a_{i+1}\ldots a_{j-1}na_{j+1}\ldots}(h_{a_im})\right] +T_{\ldots a_{j-1}ma_{j+1}\ldots}(D) - T_{\ldots a_{j-1}}{}_m{}_{a_{j+1}\ldots} = 0. \label{N rank tensr analysis step 3}
\end{align}
Letting $m:=a_j$ and defining the following notations:
\begin{align}
 T_{\ldots a_{i-1}}{}^n{}_{a_{i+1}\ldots a_{j-1}na_{j+1}\ldots} &:=\mathrm{Tr}_{(i, j)}(T_{a_1 \ldots a_N}), \notag\\[5pt]
 T_{\ldots a_{i-1} a_j a_{i+1}\ldots a_{j-1}a_ia_{j+1}\ldots} &:= (i, j)\cdot T_{a_1 \ldots a_N},
\end{align}
where the action of $\mathrm{Tr}_{(i, j)}$ contracts the indices of $a_i$ and $a_j$ lowering the rank by 2, and $(i, j)\cdot$ swaps the indices $a_i$ and $a_j$ corresponding to a transposition of the indices. Then, \eqref{N rank tensr analysis step 3} becomes 
\begin{align}
 T_{a_1 \ldots a_N}(D - 1) + \sum_{\substack{i = 1 \\ i \ne j}}^N\left[ (i, j)\cdot T_{a_1\ldots a_N}\right] - \sum_{\substack{i = 1 \\ i \ne j}}^N\left[ \mathrm{Tr}_{(i, j)}(T_{a_1 \ldots a_N}) h_{a_ia_j} \right] = 0. \label{N rank tensr analysis - specific j step 4}
\end{align}

\par To use all the symmetry available in this construction, we evaluate \eqref{N rank tensr analysis - specific j step 4} for all $j \in 1\ldots N$ and add them up to get the following:
\begin{align}
 T_{a_1 \ldots a_N}(D - 1)N + \sum_{\substack{i \ne j}}^N\left[ (i, j)\cdot T_{a_1\ldots a_N}\right] - \sum_{\substack{ i \ne j}}^N\left[ \mathrm{Tr}_{(i, j)}(T_{a_1 \ldots a_N}) h_{a_ia_j} \right] = 0, \label{N rank tensr analysis - summed j step 5}
\end{align}
where the summations now run over all $i$ and $j$ with $i \neq j$.
Because the second and third terms are symmetric under the exchange of labels $i$ and $j$, \eqref{N rank tensr analysis - summed j step 5} can be also written as,
\begin{align}
 T_{a_1 \ldots a_N}(D - 1)N + 2\sum_{\substack{i < j}}^N\left[ (i, j)\cdot T_{a_1\ldots a_N}\right] - 2\sum_{\substack{ i < j}}^N\left[ \mathrm{Tr}_{(i, j)}(T_{a_1 \ldots a_N}) h_{a_ia_j} \right] = 0. \label{N rank tensr analysis - summed j step 6}
\end{align}
\subsubsection{Tools From Representation Theory}
\par To be able to proceed, we will make use of the tools of representation theory. More specifically, we will consider the collective action of $GL(V)$ (the general linear group of the vector space V consisting of all the linear invertible transformations from V to V) and $S_N$ (the symmetric group of N elements, i.e., the group of all possible permutations of N elements which we assume to be finite) on $W = V^{\otimes N}$ (the vector space of rank-$N$ tensors), with the constituent vector space $V$ defined over the field $\mathbb{F}$, e.g., $\mathbb{C}$ or $\mathbb{R}$, which is to be specified later. As a side note, in the terminology of representation theory, $W$ is referred as a $\mathbb{F}\left(S_N \times GL(V)\right)$ module, in other words, the vector space over the field $\mathbb{F}$ on which the representations of the group elements of $GL(V)$ and $S_N$ act accordingly (see Chapter~4 in \cite{James-Gordon}).

\par In particular, for a given basis $\{\hat{e}_\mu\}$ for the vector space $V$, $T \in W = V^{\otimes N}$ is given by 
$T = T^{\mu_1 \ldots \mu_N} \left(\hat{e}_{\mu_1} \otimes \ldots \otimes \hat{e}_{\mu_N}\right)$. Then, the action of $\sigma \in S_N$ on $T$ permutes the indices of the basis vectors:
\begin{align}
 \sigma \cdot T &= T^{\mu_1 \ldots \mu_N} \left(\hat{e}_{\mu_{\sigma(1)}} \otimes \ldots \otimes \hat{e}_{\mu_{\sigma(N)}}\right)=T^{\mu_{\sigma^{-1}(1)} \ldots \mu_{\sigma^{-1}(N)}}\left(\hat{e}_{\mu_1} \otimes \ldots \otimes \hat{e}_{\mu_N}\right), \label{Action of S_N on T}
\end{align}
where the second equality follows by defining new dummy index labels as $\sigma(i) = j, j = \sigma^{-1}(i)$, thus $\mu_i \xrightarrow{} \mu_{\sigma^{-1}(j)}$ in the implicit summation (Einstein's summation convention). On the other hand, $g \in GL(V)$ acts on each $\hat{e}_{\mu_i}$, that is:
\begin{align}
 g \cdot T = T^{\mu_1 \ldots \mu_N} \left(g(\hat{e}_{\mu_1}) \otimes \ldots \otimes g(\hat{e}_{\mu_N})\right).
\end{align}
In fact, these two actions commute and the vector space $W = V^{\otimes N}$ decomposes as follows:
\begin{equation}
 V^{\otimes N} \cong \bigoplus_{\lambda \vdash N}(S_\lambda V)^{\oplus m_\lambda}, \label{Schur--Weyl form 1}
\end{equation}
which is called the Schur--Weyl Duality (see p.~77 in \cite{Fulton-Harris}). \par Here, $\lambda = (\lambda_1, \ldots, \lambda_k)$ is a partition of N, i.e. $\sum_i \lambda_i = N$ denoted by $\lambda \vdash N$, and corresponds to a Young tableau of shape $\lambda$ ($\lambda - tableau$) characterizing the irreducible representations of the symmetric group $S_N$ (see Chapter 2.1 of \cite{Bruce}).

$S_\lambda{V}$ is the image of the action of $\hat{\alpha}_\lambda$, Young symmetrizer on $V^{\otimes N}$ ($\hat{\alpha}_\lambda$ is an element in $\mathbb{F}S_N = \left\{ \sum_{\sigma \in S_N} c_\sigma \cdot \sigma \,\middle|\, c_\sigma \in \mathbb{F}, \sigma \in S_N \right\}$ whose action on a tensor $T \in V^{\otimes N}$ projects $T$ into $T^{(\lambda)}$ living in the irreducible module of $GL(V)$, named the Weyl module corresponding to $\lambda$.) (see pp.~76--78 in \cite{Fulton-Harris}). Essentially, a given Weyl module is a subspace of $V^{\otimes N}$ consisting of tensors respecting the symmetries of the corresponding $\lambda - tableau$.
\par On the other hand, $m_\lambda$ is the dimension of the irreducible module of $S_N$ corresponding to the partition $\lambda$, called Specht module and denoted by $S^\lambda$ (see Chapter~2.3 in \cite{Bruce}).
\par While the Schur--Weyl Duality demonstrates the duality between the representations of $GL(V)$ and $S_N$, the form in \eqref{Schur--Weyl form 1} underlies the action of $GL(V)$, whereas the following isomorphic decomposition signifies the action of both $GL(V)$ and $S_N$; more useful for our case (see Corollary~4.59 in \cite{Pavel} and Schur--Weyl Duality~3.4 in \cite{Stevens}).
\begin{align}
 V^{\otimes N} \cong \bigoplus_{\lambda \vdash N}S^{\lambda}\otimes(S_\lambda V), \label{Schur--Weyl form 2}
\end{align}
with $\dim(S^\lambda) = m_\lambda$.
\par Regarding the field $\mathbb{F}$ that the vector space V is defined over, the construction of the Specht and Weyl modules to build the Schur--Weyl Duality is generally carried out over $\mathbb{F} = \mathbb{C}$ which is the case in \cite{Fulton-Harris}, \cite{Bruce}, \cite{Stevens}. However, this is not a requirement. In fact many of the properties of the irreducible modules considered here are independent of the field $\mathbb{F}$. Firstly, Specht modules of $S_N$ are always irreducible when the field $\mathbb{F}$ has characteristic zero but may not be when the field has a prime characteristic (see Chapter~4 in \cite{GD-James}, ~Chapter~7.1 in \cite{James-Kerber}, ~Chapter~2.4 in \cite{Bruce}) 
The same is true for the irreducibility of Weyl modules of $GL(V)$ (see ~Chapter~8.3 in \cite{James-Kerber}). Furthermore, the Schur--Weyl Duality is still valid when the field has characteristic zero (see Chapters~4.18 and~4.19 in \cite{Pavel}), and even when it has a positive characteristic~\cite{Doty}. 
\par Consequently, the decomposition in \eqref{Schur--Weyl form 2} is valid over any field $\mathbb{F}$ with characteristic zero, where $S^\lambda$ and $S_\lambda(V)$ are the Specht and Weyl modules respectively as defined previously. This point is noteworthy and necessary because in general relativity the tensor algebra is constructed over the field of real numbers, i.e. $\mathbb{F} = \mathbb{R}$. So, from now on, we restrict $\mathbb{F}$ to be a field of characteristic zero.

\par With the decomposition in \eqref{Schur--Weyl form 2}, a given tensor $T \in V^{\otimes N}$ can be expressed as:
\begin{align}
 T = \sum\limits_{\lambda \vdash N}T^{(\lambda)} = \sum\limits_{\lambda \vdash N}\left(\sum\limits_ka_k^{(\lambda)}\otimes b_k^{(\lambda)}\right), \label{Tensor decomposition by Schur--Weyl}
\end{align}
for some $a_k^{(\lambda)} \in S^\lambda$ and $b_k^{(\lambda)} \in S_\lambda(V)$ according to the very definition of direct sum and tensor product. As stated earlier, $T^{(\lambda)}$ is the constituent of the full tensor $T$ living inside the irreducible module corresponding to the partition $\lambda$, under the action of $GL(V)$ and $S_N$. Furthermore, each $T^{(\lambda)}$ are distinct and independent, thus the expression of $T$ with this decomposition is unique. This is also a consequence of the definition of direct sum (see pp.~41--43 in \cite{AdvancedAlgebra}).
\par To understand the second term in equation \eqref{N rank tensr analysis - summed j step 6}, we look at how the element
\begin{equation}
 \tau := \sum\limits_{i<j}(i,j) = \sum\limits_{\sigma \in C_2}\sigma, \label{Tau operator}
\end{equation}
in $\mathbb{F}S_N$ acts on $T$, or on a irreducible constituent $T^{(\lambda)}$. ($C_2$ is the conjugacy class of transpositions, consisting of all $\binom{N}{2}$ transpositions; since $\tau$ is essentially the sum of all possible transpositions in $S_N$). Firstly, for any $\pi \in S_N$, the following holds
\begin{align}
 \pi\tau\pi^{-1} &= \sum\limits_{i<j}\pi(i,j)\pi^{-1} =\sum\limits_{i<j}(\pi(i),\pi(j)) = \tau,
\end{align}
because the last summation still runs over all possible transpositions in $S_N$. Therefore, $\tau$ commutes with any $\pi \in S_N$; thus living in the center of $S_N$, $Z(S_N)$.

\begin{definition}
Given the G-modules V and W (vector spaces that the representations of $g \in G$ act on), G-homomorphism is a linear transformation from $V$ to $W$ respecting the action of G. That is,
\begin{align}
 &\theta: V \xrightarrow{} W,\notag\\
 &\theta(g.v) = g\theta(v),
\end{align}
with $v \in V$ and $\theta(v) \in W$.
\end{definition}

\begin{definition}
 Given the $\mathbb{F}G$-modules V and W,
 \begin{equation}
 Hom_G(V,W) = \{\theta:V \xrightarrow{} W \hspace{0.1cm}| \hspace{0.1cm} \theta(gv) = g\theta(v)\},
\end{equation}
 is the set of all G-homomorphims from V to W.
 \begin{equation}
 End_G(V) = Hom_G{(V,V)},
\end{equation}
 is the set of all G-homomorphisms from a G-module V to itself.
\end{definition}

\par In fact, for $\rho, \theta \in Hom_G(V,W)$, $v \in V$ and $\lambda \in \mathbb{F}$ we define addition and multiplication on $Hom_G(V,W)$ as follows: 
\begin{align}
 (\rho + \theta)(v) &= \rho(v) + \theta(V), \notag\\
 \lambda\theta(v) &= \theta(\lambda v).
\end{align}
This makes $Hom_G(V,W)$ a vector space over $\mathbb{F}$.

\par In our case, $\tau$ in \eqref{Tau operator}, with the group action of $S_N$ defined as in \eqref{Action of S_N on T}, is an $S_N$-homomorphism. Furthermore, since we are considering the $S_N$ homomorphism $\tau$ composed of the elements $\sigma \in S_N$, its action on $v \in V$ an $S_N$-module, returns an element in the same module V. Therefore, $\tau$ is a $S_N$-homomorphism in $End_{S_N}(V)$ for any $S_N$-module V.
\par Now, with the decomposition established in \eqref{Tensor decomposition by Schur--Weyl}, we can focus on the Specht modules $S^\lambda$, the irreducible modules of $S_N$, and utilize Schur's Lemma in our analysis.
\vspace{0.5cm}

\begin{schurs_lemma}
Given the irreducible $\mathbb{F}G$-modules V and W, if $\theta:V \xrightarrow{} W$ is a G-homomorphism, then it is either an isomorphism or the zero map.
\end{schurs_lemma}

\begin{corollary} (Schur's Lemma for algebraically closed fields) Let V be an irreducible G-module over an algebraically closed field $\mathbb{F}$, and $\phi: V \xrightarrow{} V$ be a G-homomorphism (G-endomorphism). Then, $\phi = \lambda \hat{1}$ where $\lambda$ is some scalar in $\mathbb{F}$ and $\hat{1}$ is the identity operator in GL(V) 
\textnormal{(Chapter~1.3 in \cite{Pavel})}\label{Schur corollary}.
\end{corollary}
\par The proof of the corollary follows from the fact that $\phi - \lambda \hat{1}$ is also a G-endomorphism from V to itself, and has zero determinant (not an isomorphism) assuming that $\lambda$ is an eigenvalue of the operator $\phi$ lying in the field $\mathbb{F}$ (thus the requirement of algebraic closedness). 
\par Before moving on, we also define the inner product in the space $W = V^{\otimes N}$. From now on, we specifically focus on the fields $\mathbb{R}$ and $\mathbb{C}$ as our underlying field $\mathbb{F}$. Let $\langle, \rangle_0$ be any positive definite inner product (dot product for $\mathbb{R}$ and Hermitian product for $\mathbb{C}$) and $\{\hat{e}_\mu\}$ be an orthonormal basis for V with $\langle \hat{e_i}, \hat{e_j} \rangle_0 = \delta_{ij}$. Then, the inner product $\langle, \rangle$ in $W$ defined as follows,
\begin{align}
 \langle \left(\hat{e}_{\mu_1} \otimes \ldots \otimes \hat{e}_{\mu_N}\right), \left(\hat{e}_{\nu_1} \otimes \ldots \otimes \hat{e}_{\nu_N}\right)\rangle = \prod_{i = 1}^N \langle\hat{e}_{\mu_i}, \hat{e}_{\nu_i} \rangle_0 = \prod_{i = 1}^N \delta_{\mu_i \nu_i},
\end{align}
is invariant under the action of $S_N$, according to the definition in \eqref{Action of S_N on T}. That is,
\begin{align}
 \langle \left(\hat{e}_{\mu_{\sigma(1)}} \otimes \ldots \otimes \hat{e}_{\mu_{\sigma(N)}}\right), \left(\hat{e}_{\nu_{\sigma(1)}} \otimes \ldots \otimes \hat{e}_{\nu_{\sigma(N)}}\right)\rangle &= \prod_{i = 1}^N \delta_{\mu_{\sigma(i)} \nu_{\sigma(i)}} = \prod_{j = 1}^N \delta_{\mu_j \nu_j} \notag\\
 &=\langle \left(\hat{e}_{\mu_1} \otimes \ldots \otimes \hat{e}_{\mu_N}\right), \left(\hat{e}_{\nu_1} \otimes \ldots \otimes \hat{e}_{\nu_N}\right)\rangle. \label{S_N invariant inner product}
\end{align}
Then, $\sigma^\dagger \cdot \sigma = \hat{1}$; hence any element $\pi \in S_N$ with $\pi^{-1} = \pi$ is self-adjoint, and so is any real linear combination of them in $\mathbb{R}S_N$. Fortunately, our operator $\tau$ is of this kind due to $(i,j)^{-1} = (i,j)$ for $(i,j) \in C_2$. Consequently, $\tau$ has real eigenvalues as an operator acting on $T \in W$. With this, proof of the corollary \eqref{Schur corollary} can be manipulated and applied to $\tau$ to yield: $\tau = c_{\lambda}\hat{1}_d$ on the irreducible Specht module $S^\lambda$, with $c_\lambda \in \mathbb{R}$ and $d = \dim(S^\lambda)$, i.e. it acts as a real scalar on the irreducible $S_N$ modules over $\mathbb{F} = \mathbb{C}$ or $\mathbb{R}$. 
 
\par So, the action of $\tau$ on $T^{\lambda}$ in the decomposition \eqref{Tensor decomposition by Schur--Weyl} is
\begin{align}
 \tau \cdot T^{(\lambda)} = c_\lambda T^{(\lambda)}, \quad c_\lambda \in \mathbb{R}, \label{Tau action on T_lambda}
\end{align}
as it acts on the Specht module part of the decomposition in \eqref{Schur--Weyl form 1}, without affecting the corresponding Weyl module. To calculate $c_\lambda$, we take the trace on both sides of $\tau = c_\lambda \hat{1}_d$ for the module corresponding to $\lambda$. Note that trace of a representation is called its character and denoted by $\chi$; and it solely depends on the conjugacy class that the group element belongs to and the irreducible module it acts on, so we will label it by $\lambda$ again to denote the module (see Chapter~1.8 in \cite{Bruce}),
\begin{align}
 \sum\limits_{\sigma \in C_2}\chi^{(\lambda)}(\sigma) = c_\lambda \dim(S^\lambda). \label{c_lambda eqn 1}
\end{align}
Since all the transpositions live in the same conjugacy class $C_2$, they have the same character on this module. Denoting $\chi^{(\lambda)}(\sigma) = \chi^{(\lambda)}(C_2)$ for all $\sigma \in C_2$ and noting that the size of $C_2$ is $\binom{N}{2}$, \eqref{c_lambda eqn 1} yields
\begin{align}
 c_\lambda = \binom{N}{2}\frac{\chi^{(\lambda)}{(C_2)}}{\dim(S^\lambda)}.
\end{align}
In the further analysis, what will be important is not the exact value of $c_\lambda$ but an upper bound. Luckily, there is a nice relation specifically for the case of transpositions given in Theorem~3.5 of \cite{upper_bound} as follows
\begin{align}
 \left|\frac{\chi^{(\lambda)}{(C_2)}}{\dim(S^\lambda)}\right| \leq \max\left\{\frac{\lambda_r}{N}, \frac{\lambda_c}{N}\right\},
\end{align}

\noindent where $\lambda_r$ is the length of the first row and $\lambda_c$ is the length of the first column in the corresponding Young tableau of shape $\lambda$. That is, for the partition $\lambda = (\lambda_1, \ldots, \lambda_k)$ with $\sum_{i=1}^k\lambda_i = N$, $\lambda_r = \lambda_1$ and $\lambda_c = k$. Therefore, it is maximized for the fully symmetric partition $\lambda = (N)$ with $\lambda_r = N, \lambda_c = 1$ and for the fully antisymmetric partition $\lambda = (1^N)$ with $\lambda_r = 1, \lambda_c = N$, yielding the value 1. So, for any $\lambda$ partitioning N
\begin{equation}
 |c_\lambda| \leq \binom{N}{2}. \label{bound on c_lambda}
\end{equation}
\par Based on these results, the constraint equation on any rank-N FLRW-compatible tensor, i.e. equation \eqref{N rank tensr analysis - summed j step 6}, becomes
\begin{align}
 \sum\limits_{\lambda \vdash N}\left[\left((D-1)N + 2c_\lambda\right)T^{(\lambda)}\right] - 2\sum_{\substack{ i < j}}^N\left[ \mathrm{Tr}_{(i, j)}(T_{a_1 \ldots a_N}) h_{a_ia_j} \right] = 0, \label{main constraint equation}
\end{align}
subject to \eqref{bound on c_lambda}.
\subsubsection{Construction of a Strong Induction}\label{Section: Construction of Strong Induction}
To obtain an expression for $T_{\mu_1 \ldots \mu_N}$ satisfying \eqref{main constraint equation} we will utilize a strong inductive argument. For this, initially we build on the Definition 3 in \cite{metin}.
\begin{definition}\label{definition:g-u-monomial}
The \textbf{$g$-$u$ monomial of rank $(N,r, \sigma)$} is given by
\begin{equation}
 M^{(r, \sigma)}_{\mu_1 \ldots \mu_N} = \sigma \cdot \left(g_{\mu_1\mu_2} \otimes \ldots \otimes g_{\mu_{2r-1}\mu_{2r}} \otimes u_{\mu_{2r+1}}\otimes \ldots \otimes u_N\right),
\end{equation}
which is an N-rank tensor composed of $r$ metric tensors and $N-2r$ timelike unit vectors $u^\alpha = \delta^\alpha_0$ with $N$ and $r$ nonnegative integers, $N\geq 2r$ and $\sigma \in S_N$ permuting the indices. 
\par Also, $M^{(r)}_{\mu_1 \ldots \mu_N} = M^{(r, \id)}_{\mu_1 \ldots \mu_N}$ where $\id \in S_N$ is the identity element.
\end{definition}

\begin{definition}
 Let $\mathcal{M}_{(N,r)}$ be the orbit of $M^{(r)}_{\mu_1 \ldots \mu_N}$, consisting of all the distinct $g$-$u$ monomials of rank $(N,r)$,
 \begin{equation}
 \mathcal{M}_{(N,r)} =\{M^{(r, \sigma)}_{\mu_1 \ldots \mu_N}: \sigma \in S_{(N,r)}\}, \label{maximal-distinct-set-of-g-u-monomials}
\end{equation}
 and $Stab_{(N,r)}$ be the stabilizer of $M^{(r)}_{\mu_1 \ldots \mu_N}$,
 \begin{equation}
 \Stab_{(N,r)} = \{\sigma \in S_N : M^{(r, \sigma)}_{\mu_1 \ldots \mu_N} = M^{(r)}_{\mu_1 \ldots \mu_N}\}.
\end{equation}
Finally, let $S_{(N,r)}$ be a set of representatives for the left cosets of $Stab_{(N,r)}$. It corresponds to the smallest subgroup of $S_N$ that can generate any monomial of rank $(N,r)$ from $M^{(r)}_{\mu_1 \ldots \mu_N}$.
\end{definition}
\noindent As a remark, the number of possible distinct \textit{g-u} monomials of rank $(N,r, \sigma)$, $N$ and $r$ fixed, is the size of this set, and given by the Orbit-Stabilizer Theorem as follows:
\begin{equation}
 |\mathcal{M}_{(N,r)}| = |S_{(N,r)}|= \frac{N!}{2^rr!(N-2r)!}, \label{number of distinct (N,r) monomial formula}.
\end{equation}
For instance, $\mathcal{M}_{(4,2)} = \{g_{\mu_1\mu_2}g_{\mu_3\mu_4},g_{\mu_1\mu_3}g_{\mu_2\mu_4}, g_{\mu_1\mu_4}g_{\mu_2\mu_3} \}$ with $|\mathcal{M}_{(4,2)}| = 4!/(2^2.2) = 3$.

\par Then, we state the main result of this paper rigorously as follows.
\begin{theorem}\label{theorem 1}\textnormal{\textbf{(Monomial Decomposition)}}
 In $(D+1)$-dimensional FLRW spacetime equipped with Levi-Civita connection, any N-rank FLRW compatible tensor with $D>N$ admits the decomposition:
 \begin{equation}
 T_{\mu_1 \ldots \mu_N} =\sum\limits_{r}\sum\limits_{\sigma \in S_{(N,r)}}\alpha_{\sigma r}\hspace{0.1cm} M^{(r,\sigma)}_{\mu_1 \ldots \mu_N}, \label{main theorem expansion}
\end{equation}
 for some $\alpha_{\sigma r}(t) \in \mathbb{F}$, the underlying field of the vector space ($\mathbb{C}$ or $\mathbb{R}$). 
\end{theorem}

\par It is important to emphasize that the restriction $D>N$ is not merely technical. For the connected isotropy group $SO(D)$, the spatial Levi--Civita tensor $\epsilon_{a_1\cdots a_D}$ is an invariant tensor of rank $D$. Hence, for $N\geq D$, additional parity-odd invariant structures may enter the decomposition. The assumption $D>N$ excludes these Levi--Civita-type structures and guarantees that the invariant covariant tensors are generated only by the spatial metric, equivalently by $g_{\mu\nu}$ and $u_\mu$. If invariance under the full orthogonal group $O(D)$ is imposed, parity-odd structures are excluded by reflections; nevertheless, the condition $D>N$ is kept here to make the statement independent of this distinction.
\par To prove Theorem \eqref{theorem 1} for any N (finite), we assume that it holds for any k-rank FLRW compatible tensor with $D>k$ for $k=1, 2, \ldots, N-1$. 
\par We begin by the fully spatial part and utilize the constraint equation \eqref{main constraint equation}. $\mathrm{Tr}_{(i, j)}(T_{a_1\ldots a_N})$ is a rank N-2 FLRW compatible tensor by the property \ref{FLRW Compatible statement 1}. Then, by the assumption, it is expressed as
\begin{equation}
 \mathrm{Tr}_{(i, j)}(T_{a_1\ldots a_N}) = \sum\limits_{r}\sum\limits_{\sigma \in S_{(N-2,r)}}\beta_{\sigma r}\hspace{0.1cm}\sigma \cdot M^{(r)}_{a_1\ldots a_{N}\setminus\{a_i,a_j\}}, \label{Trace term as monomial}
\end{equation}
where $\beta_{\sigma r}(t)\in \mathbb{F}$ and $a_1\ldots a_{N}\setminus\{a_i,a_j\}$ denotes $a_1\ldots a_{i-1}a_{i+1}\ldots a_{j-1}a_{j+1}\ldots a_{N}$.
\par If N is odd, $r \neq (N-2)/2$, then $M^{(r)}_{a_1\ldots a_{N}\setminus\{a_i,a_j\}}$ will have at least one $u_{a_i}$, for any r, in the expression. However, $u_{a_i} = \delta_{a_i}^0 = 0$, hence $\mathrm{Tr}_{(i, j)}(T_{a_1\ldots a_N}) = 0$. So, equation \eqref{main constraint equation} becomes:
\begin{equation}
 \sum\limits_{\lambda \vdash N}\left[\left((D-1)N + 2c_\lambda\right)T^{(\lambda)}\right] = 0.
\end{equation}
Because $T^{(\lambda)}$ are linearly independent according to the direct sum decomposition in \eqref{Schur--Weyl form 2}, it yields, for any $\lambda$:
\begin{equation}
 \left((D-1)N + 2c_\lambda\right)T^{(\lambda)} = 0.
\end{equation}
By \eqref{bound on c_lambda}, $\min\{{c_\lambda\} = -\binom{N}{2}}$, thus the coefficient is guaranteed to be non-zero by requiring $D>N$. Consequently, $T^{(\lambda)} = 0$ for any $\lambda$. Therefore, the fully spatial part of any odd rank FLRW-compatible tensor vanishes, provided that $D>N$. 
\par If N is even, all the terms except $r = (N-2)/2$ in \eqref{Trace term as monomial} vanish similarly. So,
\begin{equation}
 \mathrm{Tr}_{(i, j)}(T_{a_1\ldots a_N}) = \sum\limits_{\sigma \in S_{N-2}}\beta_{\sigma}\hspace{0.1cm}\sigma \cdot \left(h^{\otimes (N/2 -1)}\right)_{a_1 \ldots a_{N}\setminus\{a_i,a_j\}}, \label{exp with repeated monomials}
\end{equation}
for some $\beta_\sigma(t) \in \mathbb{F}$ and $(w\otimes v)_{ij} = w_i \otimes v_j $. Here, we are not using $S_{(N-2,(N-2)/2)}$ for our choice of $\sigma$ because the expressions become complicated and the repeated monomials in the expansion \eqref{exp with repeated monomials} can be always combined and the corresponding coefficients $\beta_{\sigma}(t)$ can be arranged accordingly.
Defining $\text{Fix}_{ij}(S_N) = \{\sigma\in S_N : \sigma(i) = i, \sigma(j) = j\}$, the second term in equation \eqref{main constraint equation} can be written as:
\begin{align}
 \sum_{\substack{ i < j}}^N\left[ \mathrm{Tr}_{(i, j)}(T_{a_1 \ldots a_N}) h_{a_ia_j} \right] &= \sum\limits_{i<j}\sum\limits_{\sigma \in S_{N-2}}\beta_{\sigma}\hspace{0.1cm}\sigma \cdot \left(h^{\otimes (N/2 -1)}\right)_{a_1 \ldots a_{N}\setminus\{a_i,a_j\}}h_{a_ia_j} \notag\\
 &= \sum_{\substack{ i < j}}\sum\limits_{\pi \in \text{Fix}_{ij}(S_N)}\beta_\pi\hspace{0.1cm} \pi \cdot (h^{\otimes (N/2)})_{a_1 \ldots a_{N}\setminus\{a_i,a_j\}a_ia_j}, \label{2nd term in main constrain eqn _1}
\end{align}
for some $\beta_\pi(t) \in \mathbb{F}$. Now, we can also decompose the final tensor according to \eqref{Tensor decomposition by Schur--Weyl}. As stated in \eqref{Schur--Weyl form 1}, we can express the tensor completely by its projected components into $S_\lambda V$, which is given by the action of $\hat{\alpha}_{\lambda}$, Young symmetrizer, on the tensor with $\hat{\alpha}_{\lambda} \in \mathbb{F}S_N$. So, \eqref{2nd term in main constrain eqn _1} becomes
\begin{align}
 &\sum\limits_{\lambda \vdash N}\left[\sum_{\substack{ i < j}}\sum\limits_{\pi \in \text{Fix}_{ij}(S_N)}\beta_\pi (\pi \circ \hat{\alpha}_\lambda) \cdot (h^{\otimes (N/2)})_{a_1 \ldots a_{N}\setminus\{a_i,a_j\}a_ia_j}\right] \notag\\&\hspace{3cm}= \sum\limits_{\lambda \vdash N}\sum\limits_{\rho\in S_N} \gamma_\rho^{(\lambda)} \rho\cdot (h^{\otimes (N/2)})_{a_1 \ldots a_{N}},
\end{align}
for some $\gamma_\rho^{(\lambda)}(t) \in \mathbb{F}$. Plugging back to \eqref{main constraint equation} and requiring $D>N$ yields
\begin{equation}
 T^{(\lambda)}_{a_1 \ldots a_{N}} = \frac{1}{c_\lambda + (D-1)N/2}\sum\limits_{\rho\in S_N} \gamma_\rho^{(\lambda)} \rho\cdot (h^{\otimes (N/2)})_{a_1 \ldots a_{N}}.
\end{equation}
Therefore,
\begin{align}
 T_{a_1 \ldots a_{N}} &= \sum\limits_{\lambda \vdash N}\left[\sum\limits_{\rho\in S_N} \frac{1}{c_\lambda + (D-1)N/2}\gamma_\rho^{(\lambda)} \rho\cdot (h^{\otimes (N/2)})_{a_1 \ldots a_{N}}\right] \notag\\
 &=\sum\limits_{r}\sum\limits_{\rho \in S_N}\alpha_{\rho r}\rho \cdot M^{(r)}_{a_1 \ldots a_N}, \hspace{3cm} \hspace{-2.5cm} \alpha_{\rho r} = \delta_{r,(N/2)}\sum\limits_{\lambda \vdash N} \frac{1}{c_\lambda + (D-1)N/2}\gamma_\rho^{(\lambda)},
\end{align}
which is exactly in the form (for the fully spatial part) that we want, i.e. a linear combination of \textit{g-u} monomials of rank $(N,r)$.

\par As for the mixed components (not fully spatial and not fully temporal part of the tensor T), these can be obtained by contracting T with the foliation vector $u^\alpha = \delta^\alpha_0$. For instance,
\begin{equation}
 T_{\mu_1\ldots \mu_{k-1}0\mu_{k+1}\ldots\mu_N} = T_{\mu_1\ldots \mu_N}\delta^{\mu_k}_0 =T_{\mu_1\ldots \mu_N}u^{\mu_k}. \label{conraction with u}
\end{equation}
By construction, $T_{\mu_1\ldots\mu_N}$ is FLRW compatible and so does the foliation vector $u^{\mu_k}$. Therefore, by the property \ref{FLRW Compatible statement 1}, the resulting tensor in \eqref{conraction with u} is also FLRW compatible. So, any component of T having $M$ many temporal indices and the remainder either temporal or spatial, is obtained by contracting the tensor with $M$ many foliating vectors $u^\alpha$. Since the resulting component is also an FLRW compatible tensor of rank lower than N, our inductive argument guarantees this to be expressed as a linear combination of all possible \textit{g-u} monomials of rank $(N-M,r)$. Eventually, we can calculate all possible mixed components of the tensor to be a linear combination of \textit{g-u} monomials of corresponding rank. In fact, it is easy to check that in this process we do not miss a component of the tensor T. Letting an index be labeled either as temporal or spatial, it has two possibilities. So, in total we have $2^N$ combinations for the mixed components. In this, the ones that have $M \neq 0$ temporal indices make up $\binom{N}{M}$ of $2^N$, each of which is handled by our inductive argument. The $M = 0$ case corresponds to the fully spatial component, which we calculated to be in the form we want by utilizing the symmetry constraint equation \eqref{main constraint equation}. Then, these give $\sum\limits_{M=1}^N\binom{N}{M} + 1$ many components, which is nothing but $\sum\limits_{M=0}^N\binom{N}{M} = 2^N$. So, it is verified. 
\par In the end, having all the mixed components of T calculated this way, the full form $T_{\mu_1\ldots \mu_N}$ can be re-constructed from these, by taking the tensor product with $u^\alpha$ $M$ many times (reaching the rank $N$) and by taking the resulting $N$-rank tensor's appropriate linear combinations, just like the case for the second-rank tensor, yielding again a linear combination of \textit{g-u} monomials of rank N. Finally, since we have verified the base cases $N=0, 1$ and demonstrated that for all $N\geq 2$ it follows from the truth of all previous cases $1, \ldots N-1$, the statement holds for all finite $N \geq 2$ as well by strong induction.
\par In the proof, we did not put much effort to calculate explicitly 
whether the expansion coefficients in the final expression will be only time dependent or not. Nonetheless, we will show that these coefficients are uniquely determined by the scalars obtained from the contractions of $g_{\alpha\beta}$ and $u_\alpha$ with the tensor, yielding FLRW compatible scalars.

\subsubsection{The Expansion Coefficients}
\indent \par To show that the expansion coefficients in the expression \eqref{main theorem expansion}, i.e. $\alpha(t)$, are, in fact, only time dependent; we shall accept the result of Theorem \eqref{theorem 1} for general $\alpha \in \mathbb{F}$ (which was implicitly proved in the previous section) and derive their form in terms of the tensor that is expanded.
\par First of all, we can equivalently use \textit{h-u} monomials of rank $(N,r, \sigma)$ instead of those \textit{g-u} in the expansion \eqref{main theorem expansion}, because $h_{\alpha\beta} = g_{\alpha\beta} + u_{\alpha\beta}$ and the scalars $\alpha$ can be re-arranged accordingly. Then, we can utilize the following results.

\begin{Lemma}\label{Lemma: orthogonality of h-u monomials}
 Let $K^{(r,\sigma)}_{\mu_1 \ldots \mu_N}$ and $K^{(r',\pi)}_{\mu_1 \ldots \mu_N}$ denote \textit{h-u} monomials of rank $(N,r,\sigma)$ and $(N,r',\pi)$ respectively. Then, 
 \begin{equation}
 K_{\mu_1\ldots \mu_N}^{(r,\sigma)}{K^{(r',\pi)}}_{}^{\mu_1\ldots \mu_N} = c(t)\delta_{r,r'},
\end{equation}
where $c(t)$ is some FLRW-compatible scalar. 
\par Additionally, if $r = r'$, the contraction yields non-zero only when $\sigma$ and $\pi$ send the index set of $u$ constituents to the same index set. In other words, the $u^{\otimes(N-2r)}$ constituent of both monomials occupy the same set of indices.
\end{Lemma}
\begin{proof}
 The case $r \neq r'$ is obvious, because in such scenario there will be an overlapping contraction between $u_{\mu_i}$ and $h_{\mu_i\mu_j}$ yielding zero by the orthogonality relation $h_{\alpha\beta}u^\alpha = 0$. On the other hand, when $r=r'$, the contraction will yield an FLRW-compatible scalar by property \ref{FLRW Compatible statement 1}. In fact, the contraction yields non-zero only when $r = r'$ and both have their $u$ constituents occupy the same set of indices. Otherwise, a similar overlap between $u_{\mu_i}$ and $h_{\mu_i\mu_j}$ will return zero. For example, $h_{\mu_1\mu_2}h_{\mu_3\mu_4}u_{\mu_5}u_{\mu_6}$ is not orthogonal to $h_{\mu_1\mu_3}h_{\mu_2\mu_4}u_{\mu_5}u_{\mu_6}$ but, it is orthogonal to $h_{\mu_1\mu_2}h_{\mu_3\mu_5}u_{\mu_4}u_{\mu_6}$. Because the $u^{\otimes2}$ constituent occupies the index set $\{5,6\}$ in the first, and $\{4,6\}$ in the latter.
\end{proof}
\begin{Lemma} \label{Lemma: linear independence of h-u monomials}
 $\{K^{(r,\sigma)}_{\mu_1 \ldots \mu_N} : \sigma \in S_{(N,r)}, r \leq N/2 \}$ is linearly independent only if $D\geq (N-1)/2$.
\end{Lemma}
\begin{proof}
 $r\neq r'$ case follows from the first part of Lemma \ref{Lemma: orthogonality of h-u monomials}. 
 \par If $r = r'$, we have two cases: $u^{\otimes(N-2r)}$ constituent has the same set of indices or not. If the set differs, then they are linearly independent as they are orthogonal by the second part of Lemma \ref{Lemma: orthogonality of h-u monomials}. If the sets are the same, then the $u^{\otimes(N-2r)}$ constituents are equivalent since the ordering of $u$ constituents does not matter. Then, denoting the index sets occupied by $u^{\otimes(N-2r)}$ and $h^{\otimes(r)}$ constituents by $\mathcal{U}\subset \{1,\ldots,N\}$ and $\mathcal{H} = \{1,\ldots,N\}\backslash \mathcal{U}$ respectively, it suffices to show that $\left(h^{\otimes r}\right)_{\mu_{\sigma^{-1}(1)} \ldots \mu_{\sigma^{-1}(2r)}}$ constituent of the monomials
 are linearly independent for different $\sigma \in \mathcal{S}$, 
 where $\mathcal{S} = \{\sigma \in S_{(N,r)} : \sigma\left(\{1,\ldots,2r\}\right) = \mathcal{H}\}$. Therefore, we need to show that the equation
 \begin{equation}
 \sum_{\sigma \in \mathcal{S}}c_{\sigma}\left(h^{\otimes r}\right)_{\mu_{\sigma^{-1}(1)} \ldots \mu_{\sigma^{-1}(2r)}} = 0, \label{linear independence eqn to prove}
\end{equation}
 is satisfied only when $c_{\sigma} = 0$ for all $\sigma \in \mathcal{S}$.
 \par Given an element $\tau \in \mathcal{S}$, we consider the following component in the tensor equation \eqref{linear independence eqn to prove}:
 \begin{align}
 \mu_{\tau^{-1}(1)} &= \mu_{\tau^{-1}(2)} = 1 \notag\\
 \mu_{\tau^{-1}(3)} &= \mu_{\tau^{-1}(4)} = 2 \notag\\
 &\vdots \notag\\
 \mu_{\tau^{-1}(2r-1)} &= \mu_{\tau^{-1}(2r)} = r \notag,\\
\end{align}
 which is possible if $D \geq r$. But, since $h$ is a diagonal matrix, only the term with $c_{\tau}$ remains in \eqref{linear independence eqn to prove} when evaluated at this component, yielding $c_{\tau} = 0$. Repeating this for each $\sigma \in \mathcal{S}$ gives all $c_{\sigma} = 0$, thus they are linearly independent. On the other hand, if $D<r$, at least two different pairs in the above setting will be assigned the same index. As a result, at least two different monomials in equation \eqref{linear independence eqn to prove} will not vanish, giving a non-trivial relationship between the corresponding coefficients. In fact, any choice of component where at least two pairs are assigned the same index will either return the identical constraint equation (if the number of pairs assigned the same index is kept) or return a less restrictive constraint equation (if the number of pairs assigned the same index is increased). Consequently, the coefficients will always have freedom (never restricted to be zero), thus the monomials are linearly dependent. To illustrate, we can consider $r=3$ labeled monomials in the expansion of a rank-6 tensor with $D = 2$. If we let the component be $\mu_{\tau^{-1}(1)} = \mu_{\tau^{-1}(2)} = i$, $\mu_{\tau^{-1}(3)} = \mu_{\tau^{-1}(4)} = j$, $\mu_{\tau^{-1}(5)} = \mu_{\tau^{-1}(6)} = j$ with $i,j \in {1,2}$ and $i \neq j$; \eqref{linear independence eqn to prove} gives:
 \begin{equation}
 h_{ii}\left( c_1h_{jj}h_{jj} + c_2h_{jj}h_{jj} + c_3h_{jj}h_{jj}\right) = 0,
\end{equation}
 yielding $(c_1 + c_2 + c_3 = 0)$. In fact, if we were to set all indices equal to $i \in \{1,2\}$, \eqref{linear independence eqn to prove} would give $\sum_{m=1}^{15}c_m = 0$. The latter is clearly less restrictive on the coefficients.
 \par Finally, if we require the whole set $\{K^{(r,\sigma)}_{\mu_1 \ldots \mu_N} : \sigma \in S_{(N,r)}, r \leq N/2 \}$ to be linearly independent, the above constraint on $D$ must hold for any $r \leq N/2$. Therefore, we must have $D\geq (N-1)/2$, exploiting the fact that $r$ and $D$ are integers.
\end{proof}
Returning to the calculation of the expansion coefficients, we group the monomials in the expansion according to their $r$ label and index set occupied by $u^{\otimes(N-2r)}$. For a rank-4 tensor, for instance, the expansion is expressed as follows.
\begin{align}
 T_{\mu_1\mu_2\mu_3\mu_4} = &+[ \alpha_{01}u_{\mu_1}u_{\mu_2}u_{\mu_3}u_{\mu_4}]\notag\\ 
 & + \left[ \begin{array}{c}
 \alpha_{11}h_{\mu_1\mu_2}u_{\mu_3}u_{\mu_4} + \alpha_{12}h_{\mu_1\mu_3}u_{\mu_2}u_{\mu_4} \notag +\alpha_{13}h_{\mu_2\mu_3}u_{\mu_1}u_{\mu_4} \\ +\alpha_{14}h_{\mu_2\mu_4}u_{\mu_1}u_{\mu_3} \notag +\alpha_{15}h_{\mu_1\mu_4}u_{\mu_2}u_{\mu_3} +\alpha_{16}h_{\mu_3\mu_4}u_{\mu_1}u_{\mu_2}
 \end{array}\right] \notag\\
 &+\left[\alpha_{21}h_{\mu_1\mu_2}h_{\mu_3\mu_4} +\alpha_{22}h_{\mu_1\mu_3}h_{\mu_2\mu_4} +\alpha_{23}h_{\mu_2\mu_3}h_{\mu_1\mu_4} \right].
\end{align}
Similarly, the $r=2$ labeled monomials in the expansion of an rank-5 tensor has 5 distinct index sets occupied by $u^{\otimes 1}$.
\begin{align}
 \left[\begin{array}{c}
 \left(\alpha_{211}u_{\mu_1}h_{\mu_2\mu_3}h_{\mu_4\mu_5} + \alpha_{212}u_{\mu_1}h_{\mu_2\mu_5}h_{\mu_4\mu_3} +\alpha_{213}u_{\mu_1}h_{\mu_2\mu_4}h_{\mu_3\mu_5}\right) \\
 +\left(\alpha_{221}u_{\mu_2}h_{\mu_1\mu_3}h_{\mu_4\mu_5} + \alpha_{222}u_{\mu_2}h_{\mu_1\mu_5}h_{\mu_4\mu_3} +\alpha_{223}u_{\mu_2}h_{\mu_1\mu_4}h_{\mu_3\mu_5}\right) \\
 +\left(\alpha_{231}u_{\mu_3}h_{\mu_2\mu_1}h_{\mu_4\mu_5} + \alpha_{232}u_{\mu_3}h_{\mu_2\mu_5}h_{\mu_4\mu_1} +\alpha_{233}u_{\mu_3}h_{\mu_2\mu_4}h_{\mu_1\mu_5}\right) 
 \\
 +\left(\alpha_{341}u_{\mu_4}h_{\mu_2\mu_1}h_{\mu_3\mu_5} + \alpha_{342}u_{\mu_4}h_{\mu_2\mu_5}h_{\mu_3\mu_1} +\alpha_{343}u_{\mu_4}h_{\mu_2\mu_3}h_{\mu_1\mu_5}\right)
 \\
 +\left(\alpha_{351}u_{\mu_5}h_{\mu_2\mu_1}h_{\mu_4\mu_3} + \alpha_{352}u_{\mu_5}h_{\mu_2\mu_3}h_{\mu_4\mu_1} +\alpha_{353}u_{\mu_5}h_{\mu_2\mu_4}h_{\mu_1\mu_3}\right)
 \end{array}\right].
\end{align}
As a reminder, the number of terms grouped by $r$ label is given by the formula in \eqref{number of distinct (N,r) monomial formula}. 
\par Now, we calculate the contraction of $T_{\mu_1\ldots\mu_N}$ with each of the \textit{h-u} monomial on right-hand side, yielding a system of linear equations on the coefficients $\alpha_{r\sigma}$. Then, the corresponding coefficient matrix for this system of linear equations becomes a Gram matrix, whose entries are given by $G_{ij} = K_{\mu_1\ldots \mu_N}^{(r,\sigma)_i}{K^{(r,\sigma)_j}}_{}^{\mu_1\ldots \mu_N}$ where $i,j \in \mathbb{N}$ labels the pair $(r,\sigma)$ with $\sigma \in S_{(N,r)}, r \leq N/2$. Thus, $i,j \in \{1,2\ldots, \sum_rS_{(N,r)}\}$ by \eqref{number of distinct (N,r) monomial formula}. In fact, if we order the coefficients according to their $r$ label, the coefficient matrix takes a block-diagonal form.

\par To make the construction more concrete, we provide the matrix form of the system of linear equations on the coefficients of the monomial expansion of rank-3 and rank-4 FLRW compatible tensors as examples.
\begin{align}
\left(
\begin{array}{cc}
1 & \\
&\left(\begin{array}{ccc}
 D & & \\
 & D & \\
 & & D
\end{array}\right)_{3\times 3}
\end{array}\right)
\left(\begin{array}{c}
\alpha_0 \\ \alpha_{11} \\ \alpha_{12} \\ \alpha_{13}
\end{array}\right) = -\left(\begin{array}{c}
T_{000}\\ T_m{}^m{}_0 \\ T_{m0}{}^m\\ T_{0m}{}^m
\end{array}\right),
\end{align}

\begin{align}
\left(
\begin{array}{ccc}
1 & \\
&\left(\begin{array}{ccc}
 D & & \\
 & \ddots & \\
 & & D
\end{array}\right)_{6\times 6}\\
&&\left(\begin{array}{ccc}
 D^2 & D & D \\
 D & D^2 & D \\
 D & D & D^2
\end{array}\right)_{3\times3}
\end{array}\right)
\left(\begin{array}{c}
\alpha_0 \\ \alpha_{11} \\ \vdots \\ \alpha_{16}\\\alpha_{21}\\\alpha_{22}\\\alpha_{23}
\end{array}\right) = \left(\begin{array}{c}
T_{0000}\\ T_m{}^m{}_{00} \\ \vdots\\ T_{00m}{}^m\\T_m{}^m{}_n{}^n\\T_m{}^n{}_m{}^n\\T_m{}^n{}_n{}^m
\end{array}\right).
\end{align}

\par We went through the effort of understanding the orthogonality properties of monomials to finally illustrate that the entries of the constant vector on right-hand side are always FLRW compatible scalars by property \ref{FLRW Compatible statement 1}, and the coefficient matrix, which is a Gram matrix, is always invertible since the monomials are linearly independent by Lemma \ref{Lemma: linear independence of h-u monomials} assuming that $D\geq (N-1)/2$ (see Theorem~7.2.10 in \cite{matrix_analysis}). Then, with the assumption $D>N$ of Theorem \ref{theorem 1} satisfying the requirement for Lemma \ref{Lemma: linear independence of h-u monomials}, the coefficients are uniquely determined by the contractions of the tensor, i.e. the entries of right-hand side, and all depend only on time.

\section{Applications to Modified Theories of Gravity} \label{section: applications}

\par We now present the consequences of Theorem \ref{theorem 1}. First, we derive the tensorial consequence for Einstein-like rearrangements of modified gravitational field equations. Second, we apply the resulting effective perfect-fluid form to generalized Friedmann equations and to the geodesic deviation equation. The model-dependent examples discussed afterward should be understood as illustrations of this general tensorial mechanism rather than as separate assumptions entering the proof of the theorem.
\par We stress that all results in this section are background-level statements on FLRW-compatible fields. They constrain the tensorial form of the field equations, but they do not determine the dynamics without specifying the action and the additional field equations of the chosen theory. Consequently, the universal perfect-fluid form should be interpreted as tensorial universality at the FLRW level, not as full dynamical equivalence among different modified theories.

\par Before moving on, we remind ourselves of some crucial relations which will simplify computation.

\subsection{Review of FLRW Geometry} \label{section: review of FLRW geometry}
\par FLRW metric in $D+1$ dimensions is given by,
\begin{equation}
 ds^2 = -dt^2 +R^2(t)\left( \frac{dr^2}{1-kr^2} + r^2d\Omega_{D-1}^2 \right), \label{FLRW metric in D+1}
\end{equation}
where $R(t)$ is the time dependent scale parameter, $k$ is the spatial curvature and $d\Omega^2_{D-1}$ is the infinitesimal line element of a $D-1$ dimensional unit sphere. \par In this geometry, the metric accepts the foliation $g_{\alpha\beta} = -u_\alpha u_{\beta} + h_{\alpha\beta}$ with $h_{\alpha\beta}$ denoting the induced metric of the spatial hypersurfaces which are homogeneous and isotropic, and orthogonal to the time-like foliation vector $u_{\alpha} = -\delta^0_\alpha$. By this foliation, covariant derivative is decomposed as,
\begin{equation}
 \nabla_\beta = -u_\beta u^\alpha\nabla_\alpha + h^\alpha_\beta \nabla_\alpha.
\end{equation}

\par Because the foliation vector is along the time-direction, i.e. $u^\alpha = \delta^\alpha_0$, the induced metric of the hypersurfaces $h_{\alpha\beta}$ vanishes for the temporal part, leaving $h_{ij}$ only. Furthermore, by \eqref{theorem 1}, scalar quantities are exclusively time-dependent. Thus, 
\begin{equation}
 \nabla_\alpha \phi = -u_{\alpha}\dot{\phi}, \label{covariant derivative on scalar}
\end{equation}
for any scalar quantity $\phi$ ($\dot{\phi}$ denoting the time derivative).
\par As for the vector fields, we note the following relation valid for FLRW geometry (no shear and rotation)
\begin{equation}
 \nabla_{\alpha}u_\beta = \frac{1}{D}\theta h_{\alpha\beta}. \label{covariant derivative on vector}
\end{equation}
Here, $\theta$ is the expansion parameter associated to the hypersurfaces, proportional to Hubble's parameter $H(t)$ by $\theta = DH(t)$ (See Chapter~2 in \cite{poisson}).
\par Together with metric compatibility $\nabla_\alpha g_{\mu\nu} = 0$, higher order derivatives of any dynamical tensor field can be computed easily by using \eqref{covariant derivative on scalar} and \eqref{covariant derivative on vector} only, based on the decomposition in \eqref{theorem 1}. Some useful relations are given below.
\begin{align}
\nabla_\mu u^\mu &= DH \notag\\ 
 \nabla_\mu\nabla_\nu\phi &= -h_{\mu\nu}H\dot{\phi} +u_\mu u_\nu\ddot{\phi} \notag\\
 \nabla_{\mu}\nabla^{\mu}\phi &= \square\phi = -DH\dot{\phi} - \ddot{\phi} \notag\\
 h^{\mu\nu}\nabla_\mu\nabla_\nu &= -DH\dot{\phi} \notag\\
 u^\mu u^\nu\nabla_\mu\nabla_\nu\phi &= \ddot{\phi}. \label{useful properties in FLRW}
\end{align}

\par Finally we note the perfect matter fluid energy momentum tensor, which is the source term for the usual spatially homogeneous and isotropic cosmological solution.
\begin{align}
 T_{\alpha\beta} = (\rho^{(M)} + p^{(M)} )u_\alpha u_\beta + p^{(M)}g_{\alpha\beta},
\end{align}
where $\rho^{(M)}$ and $p^{(M)}$ denotes the energy density and pressure associated to the matter distribution, with 
\begin{align}
 T_{\alpha\beta}u^\alpha u^\beta &= \rho^{(M)}, \notag\\
 T_{\alpha\beta}h^{\alpha\beta} &= Dp^{(M)}, \notag\\
 T_{\alpha\beta}g^{\alpha\beta} &:= T = Dp^{(M)} - \rho^{(M)}.
\end{align}

\subsection{Field Equations}\label{Section: field equations of generalized GR in cosmology}

\par We will start by calculating the purely geometrical relations for FLRW in $D+1$ dimensions, that is independent of the theory of Gravity.
\par For the metric \eqref{FLRW metric in D+1}, Weyl Tensor vanishes and Ricci Tensor and Ricci Scalar have the following simple forms, \cite{metin}
\begin{align}
 R_{\alpha\beta} &= g_{\alpha\beta}(D\rho_1 - \rho_2) + u_{\alpha}u_\beta(D-1)\rho_2,\notag\\
 R &= D((D-1)\rho_1 - 2\rho_2), \label{Ricci tensor & Ricci Scalar in FLRW in D+1 (GEOMETRY)}
\end{align}
where, 
\begin{align}
 \rho_1 = H^2 + \frac{k}{R^2}, \hspace{0.5cm} 
 \rho_2 = -\dot{H} + \frac{k}{R^2}, \hspace{0.5cm} 
 H = \frac{\dot{R}}{R},
\end{align}
and the Riemann Tensor has the usual form (with Weyl vanishing),
\begin{equation}
 R_{\alpha\beta\mu\nu} = \frac{1}{(D-1)}\left( g_{\alpha\mu}R_{\beta\nu} - g_{\alpha\nu}R_{\beta\mu} + g_{\beta\nu}R_{\alpha\mu} - g_{\beta\mu}R_{\alpha\nu} \right) - \frac{1}{D(D-1)}R\left( g_{\alpha\mu}g_{\beta\nu} - g_{\alpha\nu}g_{\beta\mu} \right). \label{Riemann Tensor in FLRW in D+1}
\end{equation}
\par This concludes the purely geometrical quantities needed to set the left-hand side of modified Einstein field equations. We now discuss the tensorial form of the metric equations for theories whose field equations admit an Einstein-like rearrangement. We consider the following broad class of actions introduced in \eqref{generalized action}
\begin{equation}
 S = \int d^{(D+1)}x\sqrt{-g}\left\{ \mathcal{L}_{matter} +\mathcal{L}_{fields}\left(g_{\mu\nu}, R^\rho_{\sigma\mu\nu}, \nabla_\alpha, A^{(i)}_{\mu_1\ldots\mu_k}, \ldots\right)\right\}. \label{generalized action_2}
\end{equation}
The Lagrangian is decomposed into the usual matter Lagrangian $\mathcal{L}_{matter}$, and a Lagrangian of the fields of the theory $\mathcal{L}_{fields}$. As before, the latter is essentially a generalized version of Einstein--Hilbert Lagrangian, potentially composed of all possible contractions of the metric tensor $g_{\mu\nu}$, Riemann Tensor $R^\rho_{\sigma\mu\nu}$, and dynamical tensor fields of the specific modified theory to be denoted by $A^{(i)}_{\mu_1\ldots\mu_k}$ (we label them by $i$ with $k \in \{0,\ldots, N\}$ being the rank), their covariant derivatives of any order and their tensor products.

\par If the metric field equations obtained by varying the action with respect to $g_{\mu\nu}$ can be rearranged into an Einstein-like form, then we write them as
\begin{equation}
 G_{\mu\nu} = fT^{(M)}_{\mu\nu} + Q_{\mu\nu}. \label{main field equation}
\end{equation}

\noindent Here $G_{\mu\nu}$ and $T^{(M)}_{\mu\nu}$ denote the Einstein tensor and the usual energy--momentum tensor, respectively. The scalar $f$ and the second-rank tensor $Q_{\mu\nu}$ are constructed from contractions of the constituents of the action and their covariant derivatives. In the Einstein--Hilbert limit, they reduce to $\frac{8\pi G}{c^4}$ and $0$, respectively. The assumption in \eqref{main field equation} is therefore not an additional dynamical restriction; it is the statement that all non-Einstein contributions have been moved to the right-hand side whenever such a rearrangement is possible.

\par If we analyze the above modified theory in FLRW background, and adopt FLRW compatibility for all dynamical fields, then by Corollary \eqref{main corollary}, $Q_{\mu\nu}$ has the perfect fluid form and $f$ is only time-dependent. Furthermore, $T^{(M)}_{\mu\nu}$ also has the perfect-fluid form. Therefore, the following decompositions are valid.
\begin{align}
 T_{\mu\nu}^{(M)} &= (\rho^{(M)} + p^{(M)})u_\mu u_\nu + p^{(M)} g_{\mu\nu}, \\[0.3cm]
 Q_{\mu\nu} := T_{\mu\nu}^{(F)} &= (\rho^{(F)} + p^{(F)})u_\mu u_\nu + p^{(F)} g_{\mu\nu}.
\end{align}

\par For scalar--vector--tensor theories, the corresponding rank-two effective-fluid structure has been established directly in \cite{SVT-FLRW}. Here, this structure is recovered as a special case of the general rank-$N$ FLRW-compatible tensor decomposition proved above, since the latter classifies the admissible tensorial building blocks before specializing to the symmetric rank-two metric equation. 

\par Furthermore, a closely related line of work establishes the
perfect-fluid character of the gravitational corrections directly for specific
classes of curvature-based Lagrangians. It was shown in \cite{PF in f(R)} that in
generalized Robertson--Walker spacetimes with divergence-free conformal curvature
the stress--energy tensor of any $f(R)$ model is of perfect-fluid type, and this
was extended in \cite{PF in higher order GR} to all higher-order corrections of
the form $F(R,\square R,\ldots,\square^k R)$. Within that class these results are
obtained without assuming that the fields inherit the background symmetry, and
they extend beyond exact Robertson--Walker geometry. The classification given
here is complementary: the admissible corrections are not restricted to
functionals of the curvature, but may involve dynamical tensor fields of
arbitrary rank together with their covariant derivatives of any order, in
$(D+1)$ dimensions. The perfect-fluid form then follows, through Corollary
\ref{main corollary}, from the rank-two sector of a general classification of
FLRW-compatible tensors rather than from the structure of the curvature
invariants, showing that the underlying mechanism is representation-theoretic and
independent of the corrections being curvature-based. The same mechanism was
formulated in \cite{PF in extended gravity} through the notion of a perfect
scalar, showing that the field equations of extended theories of gravity
naturally contain perfect-fluid terms which may be interpreted as geometric
contributions to the dark sector. A perfect scalar on a Robertson--Walker
background depends only on cosmic time, which is precisely the rank-zero case of
FLRW compatibility and the base step of the induction leading to Theorem
\ref{theorem 1}; the geometric perfect fluids obtained there are therefore the
physical counterpart, at ranks zero and two, of the structure that the present
decomposition establishes for arbitrary rank.

\par Having established the perfect fluid form, we henceforth denote $Q_{\mu\nu}$ by $T^{(F)}_{\mu\nu}$ since it can be treated as an effective energy momentum tensor of the additional fields entering the theory. Because $D+1$ decomposition of the metric applies, i.e. $h_{\alpha\beta} = g_{\alpha\beta} + u_\alpha u_\beta$, we have
\begin{align}
 \rho^{(M)} &= T_{\mu\nu}^{(M)}u^\mu u^\nu \notag, \hspace{0.7cm}
 p^{(M)} = \frac{1}{D}T_{\mu\nu}^{(M)}h^{\mu\nu},\\[0.3cm]
 \rho^{(F)} &= T_{\mu\nu}^{(F)}u^\mu u^\nu, \hspace{0.7cm} 
 p^{(F)} = \frac{1}{D}T_{\mu\nu}^{(F)}h^{\mu\nu}.
\end{align}

\noindent Now, we can combine all the terms and define a total, effective, energy momentum tensor.
\begin{align}
 T_{\mu\nu}^{(T)} = f(t)T_{\mu\nu}^{(M)} + T_{\mu\nu}^{(F)}, \label{total (effective) energy--momentum tensor}
\end{align}

\noindent thus the corresponding variables can be defined similarly:
\begin{align}
 \rho^{(T)} &= T_{\mu\nu}^{(T)}u^\mu u^\nu=f(t)\rho^{(M)} + \rho^{(F)}, \notag\\[0.2cm]
 p^{(T)} &= \frac{1}{D}T_{\mu\nu}^{(T)}h^{\mu\nu} = f(t)p^{(M)} + p^{(F)}. \label{effective densities}
\end{align}
Then, the Einstein tensor has the form
\begin{align}
 G_{\mu\nu} = R_{\mu\nu} - \frac{1}{2}Rg_{\mu\nu} = (\rho^{(T)} + p^{(T)})u_\mu u_\nu + p^{(T)} g_{\mu\nu}, \label{Einstein tensor in FLRW! (Theory description)}
\end{align} 
Taking the trace we calculate,
\begin{align}
 R &= \frac{2}{(1-D)}(Dp^{(T)} -\rho^{(T)}),\\
 R_{\mu\nu}&= (\rho^{(T)} + p^{(T)})u_\mu u_\nu +\frac{1}{(D-1)}(\rho^{(T)} - p^{(T)})g_{\mu\nu}. \label{Ricci Tensor & Ricci Scalar in FLRW in D+1 (THEORY)}
\end{align}
This relation defines the theory dependent description of the Ricci Tensor and the Ricci Scalar, which then can be compared to their geometrical descriptions given in \eqref{Ricci tensor & Ricci Scalar in FLRW in D+1 (GEOMETRY)} to yield the modified Einstein field equations, i.e., generalized Friedmann Equations.
\par We express both \eqref{Ricci tensor & Ricci Scalar in FLRW in D+1 (GEOMETRY)} and \eqref{Ricci Tensor & Ricci Scalar in FLRW in D+1 (THEORY)} in the orthogonal basis of $\{u_{\alpha}u_\beta, \hspace{0.1cm} h_{\alpha\beta}\}$ by using the decomposition $g_{\alpha\beta} = h_{\alpha\beta} - u_{\alpha}u_\beta$.
\begin{align}
 R_{\alpha\beta} &= D(\rho_2 - \rho_1)u_\alpha u_\beta + (D\rho_1 - \rho_2)h_{\alpha\beta}, \quad &\textbf{(geometry)}\notag\\
 R_{\alpha\beta} &= \frac{1}{D-1}\left( (D-2)\rho^{(T)} + Dp^{(T)} \right)u_\alpha u_\beta + \frac{1}{D-1}(\rho^{(T)} - p^{(T)})h_{\alpha\beta}. \quad &\textbf{(theory)}
\end{align}
Now, the two can be equated easily yielding 
\begin{align}
 D(\rho_2 - \rho_1) &= \frac{1}{D-1}\left( (D-2)\rho^{(T)} + Dp^{(T)}\right), \\
 (D\rho_1 - \rho_2) &= \frac{1}{D-1}\left(\rho^{(T)} - p^{(T)}\right).
\end{align}
We can pull $\rho_1$ and $\rho_2$ in terms of $\rho^{(T)}$ and $p^{(T)}$,
\begin{align}
 \rho_1 = \frac{2}{D(D-1)}\rho^{(T)}, \hspace{0.5cm}
 \rho_2 = \frac{1}{(D-1)}(\rho^{(T)} + p^{(T)}).
\end{align}
Indeed, $\rho_1$ and $\rho_2 - \rho_1$ gives the Generalized Friedmann Equations\begin{align}
 \dot{H} + H^2 &= -\frac{1}{D(D-1)}\left( (D-2)\rho^{(T)} + Dp^{(T)}\right), \notag\\
 H^2 + \frac{k}{R^2} &= \frac{2\rho^{(T)}}{D(D-1)}. \label{Generalized Friedmann Equations}
\end{align}
which reduce to the usual Friedman Equations in the limit $D=3$ and $\rho^{(T)} = \kappa\rho^{(M)}$, $p^{(T)} = \kappa p^{(M)}$.

\subsection{Geodesic Deviation Equation}\label{Section GDE in modified GR}
\par Similar to the previous section, we aim to relate the geometrical description of Geodesic Deviation Equation to the theory dependent quantities, consisting exclusively of $\rho^{(T)}$ and $p^{(T)}$ by \eqref{Ricci Tensor & Ricci Scalar in FLRW in D+1 (THEORY)}.
\par We consider a geodesic congruence of a longitudinal time-like vector field $v^\alpha$, not necessarily parallel to $u^\alpha$ of the background cosmic fluid. Let $\xi^\alpha$ denote the components of the deviation vector between the nearby geodesics living in this congruence and define $v^\alpha v_\alpha = \epsilon$, $v^\alpha u_\alpha = -E$. We want to calculate the geodesic deviation equation given by,
\begin{equation}
 \Ddot{\xi^a} = -R^a_{bcd}v^bv^d\xi^c.
\end{equation}

\noindent By \eqref{Riemann Tensor in FLRW in D+1}, we have

\begin{align}
 R^\alpha_{\beta\mu\nu}v^\beta v^\nu\xi^\mu&= \frac{1}{(D-1)}\left[ (R_{\beta\nu}v^\beta v^\nu)\xi^\alpha -(R_{\beta\mu}v^\beta\xi^\mu)v^\alpha + \epsilon(R^\alpha_\mu\xi^\mu)\right] - \frac{R}{D(D-1)}\epsilon \xi^\alpha,
\end{align}
where we used $v^\alpha \xi_{\alpha} = 0$ and $u^\alpha \xi_\alpha = 0$ \cite{GDE_orthogonality}. We can directly apply our result \eqref{Ricci Tensor & Ricci Scalar in FLRW in D+1 (THEORY)} to calculate the contractions.

\begin{align}
 R^\alpha_{\beta\mu\nu}v^\beta v^\nu\xi^\mu
 &= \frac{1}{(D-1)}\xi^\alpha(\rho^{(T)} + p^{(T)})E^2 +\frac{2}{D(D-1)}\rho^{(T)}\epsilon\xi^\alpha.
\end{align}
Therefore, \textbf{the geodesic deviation equation} has the following compact form, for a modified theory with field equations in \eqref{main field equation}.
\begin{align}
 \Ddot{\xi^\alpha} = -\frac{2}{D(D-1)}\epsilon\xi^\alpha\rho^{(T)} - \frac{1}{(D-1)}E^2\xi^\alpha(\rho^{(T)} + p^{(T)}), \label{GDE in FLRW}
\end{align}
recalling the definitions,
\begin{align}
 \rho^{(T)} = (fT^{(M)}_{\mu\nu} + Q_{\mu\nu})u^\mu u^\nu, \notag\\
 p^{(T)} = (fT^{(M)}_{\mu\nu} + Q_{\mu\nu})\frac{1}{D}h^{\mu\nu}.
\end{align}

In what follows, we derive the geodesic deviation equation for two known modified theories of gravity containing different tensor fields (in the FLRW geometry). Doing so, we will see how significantly the algebra simplifies applying the result \eqref{GDE in FLRW}.

\subsubsection{f(R,T) Theory of Gravity}
\par To illustrate the utility of the result \eqref{GDE in FLRW}, we consider a modified theory of gravity where the action is generalized to an arbitrary function of the Ricci Scalar (R) and the trace of the energy--momentum tensor (T) from Einstein--Hilbert action. Setting $\kappa$ to be 1, the action takes the form (see \cite{f(RT)_harko})
\begin{align}
 S = \frac{1}{2}\int d^{(D+1)}x\sqrt{-g}\{ f(R,T)\} +\int d^{(D+1)}x\sqrt{-g}\mathcal{L}_{m},
\end{align}
where $f$ is an arbitrary function of $R$ and $T$ and $\mathcal{L}_m$ is the usual matter Lagrangian, from which the energy--momentum tensor is defined as,
\begin{equation}
 T^{(M)}_{\alpha\beta} = g_{\alpha\beta}\mathcal{L}_m - 2\frac{\partial \mathcal{L}_m}{\partial g^{\alpha\beta}},
\end{equation}
assuming that $\mathcal{L}_m$ only depends on $g_{\mu\nu}$. 
\par Carrying out the variation over the metric tensor yields,
\begin{align}
 f_RR_{\mu\nu} - \frac{1}{2}fg_{\mu\nu} + \left( g_{\mu\nu}\square - \nabla_\mu\nabla_\nu\right)f_R = T_{\mu\nu} - f_T\left( T_{\mu\nu} + \Theta_{\mu\nu} \right), \label{Field Equation in f(R,T)}
\end{align}
where,
\begin{align}
 f_R = \frac{\partial f}{\partial R}, \hspace{1cm} f_T = \frac{\partial f}{ \partial T}, \hspace{1cm}
 \Theta_{\mu\nu} = g^{\alpha\beta} \frac{\delta T_{\alpha\beta}}{\delta g^{\mu\nu}}.
\end{align}
\par Now, to derive modified Friedmann Equations or GDE, all we need is to calculate effective energy density and pressure, that is $\rho^{(T)}$ and $p^{(T)}$ from $T_{\mu\nu}^{(T)}$ related to Einstein tensor via \eqref{Einstein tensor in FLRW! (Theory description)}. So, we need $G_{\mu\nu}$. Trace of \eqref{Field Equation in f(R,T)} gives the Ricci Scalar,
\begin{equation}
 R = \frac{1}{f_R}\left( (1-f_T)T - f_T\Theta - D\square f_R +\frac{1}{2}f(D+1) \right), \label{Ricci Scalar in f(R,T)}
\end{equation}
Thus the Einstein tensor is,
\begin{align}
 G_{\mu\nu} &= T_{\mu\nu}^{(T)} \notag\\
 &=\frac{1}{f_R}\left( (1-f_T)(T_{\mu\nu} - \frac{1}{2}g_{\mu\nu}T) + f_T(\frac{1}{2}g_{\mu\nu}\Theta - \Theta_{\mu\nu}) \right. \\ &\hspace{4cm}\left.+ \frac{1-D}{4}fg_{\mu\nu} + (\nabla_\mu\nabla_\nu + (\frac{D}{2}-1)g_{\mu\nu}\square)f_R \right),
\end{align}
from which we calculate $\rho^{(T)}$ and $p^{(T)}$ using \eqref{effective densities}.
\begin{align}
 \rho^{(T)} &= \frac{1}{2f_R}\Big[ \rho^{(M)}+Dp^{(M)} + f_T\left(\rho^{(M)} + p^{(M)}\right) +\frac{(D-1)}{2}f + D\left( \ddot{f_R} + H(D-2)\dot{f_R} \right) \Big], \notag\\
 p^{(T)} &= \frac{1}{2f_R}\Big[ p^{(M)}(2-D) + \rho^{(M)} + f_T\left(\rho^{(M)} +p^{(M)}\right) \Big. \notag\\& \hspace{3.5cm} \Big.-\frac{(D-1)}{2}f - D\left( \ddot{f_R} + H(D-2)\dot{f_R} \right) +2\left( \ddot{f_R} - H\dot{f_R} \right) \Big]. \label{effective densities in f(R,T)}
\end{align}
To derive \eqref{effective densities in f(R,T)}, we made use of the listed properties in \eqref{useful properties in FLRW} and used $\Theta_{\mu\nu} = -2T^{(M)}_{\mu\nu} + p^{(M)}g_{\mu\nu}$ for perfect fluid.
\par Ultimately, this yields the geodesic deviation equation by \eqref{GDE in FLRW}:
\begin{align}
 -\ddot{\xi^\mu} &= \frac{1}{D(D-1)f_R}\epsilon\xi^\mu\Big[ \rho^{(M)}\big( 1 + f_T \big) + p^{(M)}\left( D + f_T \right) \Big. \notag\\ &\hspace{6cm}\Big. +\frac{(D-1)}{2}f + D\left(\ddot{f_R} + H(D-2)\dot{f_R} \right) \Big] \notag\\ 
 &+ \frac{1}{(D-1)f_R}E^2\xi^\mu\left[ (1+f_T)(p^{(M)} + \rho^{(M)}) +\ddot{f_R} - H\dot{f_R} \right], \label{GDE in f(R,T)}
\end{align}

and the generalized Friedmann Equations by \eqref{Generalized Friedmann Equations}:
\begin{align}
\dot{H} + H^2 &= -\frac{1}{Df_R}\Big[ \rho^{(M)} + f_T( \rho^{(M)} + p^{(M)}) -\frac{f}{2} - DH\dot{f_R} \Big], \\
H^2 &= \frac{1}{D(D-1)f_R}\Big[ \rho^{(M)} + Dp^{(M)} +f_T(\rho^{(M)} + p^{(M)}) \Big.\notag\\ &\hspace{3cm}\Big.+\frac{(D-1)}{2}f + D(\ddot{f_R} + H(D-2)\dot{f_R}) \Big] -\frac{k}{R^2}. \label{Generalized Friedmann in f(R,T)}
\end{align}

\par To confirm the validity of this treatment, we can compare against the GDE and Friedmann Equations for $f(R,T)$ in literature. 
\par In \cite{f(RT)_harko}, the Generalized Friedmann Equation is calculated for several $f(R,T)$ models in $3+1$ dimensions, although a general form is not provided. Here, we will show the matching of our calculation for one of the models presented in \cite{f(RT)_harko}, which can be demonstrated for any other models similarly. We consider, $k=0$ and $D=3$ where Friedmann Equations in \eqref{Generalized Friedmann Equations} can be re-written as
\begin{align}
 \rho^{(T)} &= 3H^2, \notag\\
 p^{(T)} &= -2\dot{H} -3H^2.
\end{align}
For the model of scalar field gravity, i.e. $f(R,T) = R + 2f(T^\phi)$ (see Eq.(54)) where $f_R = 1$ \eqref{effective densities in f(R,T)} simplifies considerably and with $p^{(M)} = \frac{1}{2}w\dot{\phi}^2 - V(\phi)$, $\rho^{(M)} = \frac{1}{2}w\dot{\phi}^2 + V(\phi)$ we attain
\begin{align}
 \rho^{(T)} &= \frac{1}{2}w\dot{\phi}^2 + V(\phi) + 2f'w\dot{\phi}^2 - f, \notag\\
 p^{(T)} &= \frac{1}{2}w\dot{\phi}^2 - V(\phi) +f.
\end{align}
This matches the result in \cite{f(RT)_harko}(see Eqn. (55) and (56)).
\par Before comparing the work in GDE for $f(R,T)$ we will check the limit case, namely $f(R)$ Gravity. By setting $f_T = 0$ we get back $f(R)$ theory of gravity, where \eqref{GDE in f(R,T)} yields:
\begin{align}
 -\ddot{\xi^\mu} &= \frac{1}{D(D-1)f_R}\epsilon\xi^\mu\Big[ \rho^{(M)} + Dp^{(M)} +\frac{(D-1)}{2}f + D\left(\ddot{f_R} + H(D-2)\dot{f_R} \right) \Big] \notag\\ 
 &+ \frac{1}{(D-1)f_R}E^2\xi^\mu\left[ p^{(M)} + \rho^{(M)} +\ddot{f_R} - H\dot{f_R} \right]. \label{GDE in f(R) limit}
\end{align}
\par In \cite{GDE_f(R)_Alvaro}, the GDE for $f(R)$ Gravity is partially expressed in terms of effective pressure and energy density as we defined in \eqref{effective densities}, (see Eq.~(33) in \cite{GDE_f(R)_Alvaro}). This is equivalent to our expression \eqref{GDE in FLRW} and subsequently to \eqref{GDE in f(R) limit}. Furthermore, the Friedmann Equations for $f(R)$ gravity provided in \cite{f(R)_Thomas} take the same form of \eqref{Generalized Friedmann Equations}. This confirms the validity of our calculations in the $f(R)$ limit.

\par Regarding the GDE in $f(R,T)$ gravity, the calculation provided in Eq.~(35) in \cite{GDE_f(RT)_bauffoi} differs from our result \eqref{GDE in f(R,T)}. However, taking the $f(R)$ limit of the result in \cite{GDE_f(RT)_bauffoi} yields a relation that also differs from what we obtained in \eqref{GDE in f(R) limit}, thus contradicts the consensus in the papers \cite{GDE_f(R)_Alvaro, f(R)_Thomas}. The source of the discrepancy appears to originate from the GDE for $f(R)$ gravity given in Eq.~(4.16) in \cite{GDE_f(R)_Alejandro} matching the problematic $f(R)$ limit of the GDE in \cite{GDE_f(RT)_bauffoi}. A closer examination of the derivation provided in \cite{GDE_f(R)_Alejandro} shows an inconsistency moving from Eq.~(4.9) to Eq.~(4.15), likely originating from missing terms in the relations between congruence vector field $v^\alpha$, foliation $u^\alpha$ and the induced metric $h_{\alpha\beta}$. Consequently, the GDE expressions proposed for f(R) gravity in \cite{GDE_f(R)_Alejandro} and for $f(R,T)$ in \cite{GDE_f(RT)_bauffoi} may require some refinement in view of this discrepancy. Therefore, our result \eqref{GDE in f(R,T)} provides an alternative formulation of GDE in $f(R,T)$ gravity with a consistent $f(R)$ limit \cite{GDE_f(R)_Alvaro, f(R)_Thomas}.

\subsubsection{Einstein-Aether Theory}
\par We consider a vector-tensor theory containing a constant norm vector field established over the whole spacetime introducing a preferred frame. It is motivated to introduce such a vector field coupled to gravity, called Aether Field $w^\alpha$, to break Lorentz symmetry which is predicted to be realized at high energies by considerations of quantum gravity.
\par As an example of such Lorentz symmetry breaking theory, we follow the model presented in \cite{Einstein_Aether_main} and use the generalized expressions from \cite{Einstein_Aether_metin}. The action is given by the most general diffeomorphism invariant form built from metric tensor and aether field, containing up to second order derivatives in the latter.
\begin{align}
 S = \frac{1}{2\kappa} \int d^{(D+1)}x \sqrt{-g}\,
\bigl[ R - 2\Lambda 
- K^{\mu\nu}{}_{\alpha\beta}\, \nabla_{\mu} w^{\alpha} \nabla_{\nu} w^{\beta}
+ \lambda (w_{\mu} w^{\mu} + \varepsilon) \bigr], \label{Aether:action}
\end{align}
where,
\begin{align}
 K^{\mu\nu}{}_{\alpha\beta}
= c_{1} g^{\mu\nu} g_{\alpha\beta}
+ c_{2} \delta^{\mu}_{\alpha} \delta^{\nu}_{\beta}
+ c_{3} \delta^{\mu}_{\beta} \delta^{\nu}_{\alpha}
- c_{4} v^{\mu} v^{\nu} g_{\alpha\beta}.
\end{align}
Here $\Lambda$ is the cosmological constant which we will omit in the following calculation, $w^\mu$ is the aether field coupling to gravity via covariant derivatives having constant norm $w^\mu w_\mu = -\epsilon$ fixed by the Lagrange multiplier $\lambda$. Timelike (null) aether field corresponds to $\epsilon = +1$ ($\epsilon = 0$) studied in \cite{Einstein_Aether_main, Einstein_Aether_metin}. The coefficients $c_1, c_2, c_3, c_4$ are dimensionless constants. Varying the action with respect to metric, aether field and the Lagrange multiplier yields the following field equations and the constraint equation respectively:
\begin{align}
 G_{\mu\nu} &= \nabla_{\alpha} \Bigl[
 J^{\alpha}{}_{(\mu} w_{\nu)}
 - J_{(\mu}{}^{\alpha} w_{\nu)}
 + J_{(\mu\nu)} w^{\alpha}
\Bigr]
+ c_{1} \bigl( \nabla_{\mu} w_{\alpha} \nabla_{\nu} w^{\alpha}
 - \nabla_{\alpha} w_{\mu} \nabla^{\alpha} w_{\nu} \bigr),
 \notag \\&\quad+ c_{4} {w'}_{\mu} {w'}_{\nu}
+ \lambda w_{\mu} w_{\nu}
- \frac{1}{2} L\, g_{\mu\nu}, \label{Aether: field eqn on metric}\\[10pt]
0 &= c_{4}\,{w'}_{\alpha} \nabla_{\mu} w^{\alpha}
+ \nabla_{\alpha} J^{\alpha}{}_{\mu}
+ \lambda w_{\mu}, \label{Aether: field eq on aether field}\\[10pt]
w^\mu w_\mu &= -\epsilon, \label{Aether: norm constraint equation}
\end{align}
where
\begin{equation}
 {w'}_\mu \equiv w^\alpha\nabla_\alpha w_\mu, \hspace{1cm} J^{\mu}{}_{\alpha} \equiv K^{\mu\nu}{}_{\alpha\beta}\, \nabla_{\nu} w^{\beta}, \hspace{1cm} L \equiv J^{\mu}{}_{\alpha}\, \nabla_{\mu} w^{\alpha}.
\end{equation}
Now imposing FLRW geometry, Theorem \eqref{theorem 1} applies and the algebra simplifies tremendously. From \eqref{general vector field result}, aether field must be along the foliation vector $u^\alpha$ to be FLRW compatible. Letting $w^\mu = \beta(t)u^\mu$, \eqref{Aether: norm constraint equation} forces $\beta^2 = \epsilon$. Then $w^\alpha$ satisfies the geodesic equation, thus $w'_\alpha = 0$. The tensor $J_{\mu\nu}$ is constructed from $w^\alpha$, therefore it is also FLRW-compatible. From \eqref{general second-rank tensor field result}, $J_{\mu\nu} = \alpha_1(t)u_{\mu}u_{\nu} + \alpha_2(t)h_{\mu\nu}$ and projecting along the foliation and the induced metric gives:
\begin{align}
 J_{\mu\nu} = -c_2\beta HDu_{\mu}u_\nu + \beta H(c_1 + c_3 + Dc_2)h_{\mu\nu}. \label{J tensor in aether}
\end{align}

Using the properties presented in section \ref{section: review of FLRW geometry}, we can calculate the following:
\begin{align}
 \nabla_\alpha J_{\mu\nu} &= u_\alpha\dot{H}\beta\big[ c_2 Du_\mu u_\nu - h_{\mu\nu}(c_1 + c_3 + Dc_2) \big] +\beta H^2(c_1 + c_3)(h_{\mu\alpha}u_\nu + h_{\alpha\nu}u_\mu), \\
 \nabla_\alpha J^\alpha{}_\mu&=u_\mu\beta D\big[ H^2(c_1 + c_3) -c_2\dot{H} \big], \notag\\
 w^\alpha(\nabla_\alpha J_{\mu\nu}) &= \beta^2\big[ \dot{H}(c_1 + c_3 + Dc_2)h_{\mu\nu} -D\dot{H}c_2u_\mu u_\nu \big], \notag \\
 L&=D(\beta H)^2(c_1 + c_3 + Dc_2), \\
 \lambda &= D\big[ c_2\dot{H}-H^2(c_1 + c_3) \big],
\end{align}
where the Lagrange multiplier $\lambda$ is attained by contracting \eqref{Aether: field eq on aether field} with $u^\alpha$. Then, Einstein tensor can be computed easily from \eqref{Aether: field eqn on metric} utilizing the symmetries: $\nabla_\alpha w_\mu = \nabla_\mu w_\alpha$ and $J_{\mu\nu} = J_{\nu\mu}$ as they are FLRW compatible.
\begin{align}
 G_{\mu\nu} = \beta^2(c_1 + c_3 + Dc_2)\Big[ -\frac{1}{2}DH^2u_\mu u_\nu + \left( \dot{H} + \frac{1}{2}DH^2\right)h_{\mu\nu} \Big]. \label{Aether: Einstein tensor}
\end{align}

\par It follows immediately that the dynamics of the theory become indistinguishable from General Relativity if the field is null, or the coupling constants satisfy $c_1 + c_3 + Dc_2 = 0$. 
The case of $\beta^2 = 0$ (i.e., $\epsilon = 0$) was previously studied in Ref. \cite{Einstein_Aether_metin}. In that work, one notes that cosmological solutions were derived for a spatially flat FLRW geometry using a null aether field ansatz $w^\mu = \phi(t)(1, 1/R(t), 0, 0)$. This vector is indeed null with respect to the FLRW metric for \(k=0\). However,
it is not itself FLRW-compatible, since it contains a preferred spatial component.
Thus, the ansatz should be understood as a null vector field on an FLRW metric,
rather than as a vector field preserving the FLRW symmetry at the level of the
field configuration. In the strictly symmetry-preserving sector considered here,
FLRW compatibility restricts the vector field to \(w^\mu=\beta(t)u^\mu\), and the
null condition then gives \(\beta(t)=0\). Therefore, the corresponding
FLRW-compatible null branch is trivial. This does not exclude the possibility
that a non-FLRW-compatible null ansatz may nevertheless lead to an effective
rank-two contribution of perfect-fluid form after substitution into the field
equations.
\par Given the triviality of the null case, we will consider $\beta^2 \equiv 1$, i.e. timelike aether field. Adding the matter to the action \eqref{Aether:action} we first calculate the effective densities from \eqref{effective densities}:
\begin{align}
 \rho^{(T)} &= -\frac{1}{2}DH^2(c_1 + c_3 + Dc_2) + \kappa\rho^{(M)}, \notag\\
 p^{(T)}&= (c_1 + c_3 + Dc_2)\left( \dot{H} + \frac{1}{2}DH^2\right) + \kappa p^{(M)}.
\end{align}
These yield the Friedmann Equations by \eqref{Generalized Friedmann Equations}:
\begin{align}
 H^2\left[ 1 - \frac{1}{D-1}(c_1 + c_3 + Dc_2) \right] &= \frac{-2\kappa}{D(D-1)}\rho^{(M)} -\frac{k}{R^2}, \notag\\
 \dot{H}\left[ 1 + \frac{1}{D-1}(c_1 + c_3 + Dc_2) \right] &=\frac{-\kappa}{(D-1)}(\rho^{(M)} + p^{(M)}) +\frac{k}{R^2}.
\end{align}
From this result, we see that the effect of aether is to renormalize the $\kappa$ factor (hence the gravitational constant) by $\left(1+\frac{1}{(D-1)}(c_1 + c_3 + Dc_2) \right)^{-1}$ and to introduce an additional perfect fluid on top of the matter content (see Section~8 in \cite{Einstein_Aether_main}).
\par We can also calculate the GDE by \eqref{GDE in FLRW}. It gives,
\begin{align}
 \Ddot{\xi^\alpha} &= \frac{\xi^\alpha}{(D-1)}
(c_1 + c_3 + Dc_2)(H^2\epsilon- E^2\dot{H}) \notag\\ \hspace{0.2cm}&- \frac{\xi^\alpha}{(D-1)}E^2\kappa\big(\rho^{(M)} + p^{(M)}\big) - 2\frac{\xi^\alpha}{D(D-1)}\epsilon\kappa\rho^{(M)}.
\end{align}
The first term is the contribution of Aether Field and the remainder correspond to the bare General Relativity sector.

\section{Conclusion}
\label{sec:conclusion}

In this work, we have given a systematic analysis of tensor fields compatible with the symmetries of FLRW spacetime. The central result is Theorem \ref{theorem 1}, which states that, in $(D+1)$ dimensions and for $D>N$, every covariant rank-$N$ FLRW-compatible tensor admits a finite decomposition into $g$-$u$ monomials, namely tensor products generated by the spacetime metric $g_{\mu\nu}$ and the unit timelike vector $u_\mu$, with scalar coefficients depending only on cosmological time. This provides a general tensorial classification of the structures allowed by spatial homogeneity and isotropy, independently of the detailed dynamics of any particular gravitational theory.

The condition $D>N$ plays an essential role in the classification. It excludes additional parity-odd invariant tensors constructed from the spatial Levi--Civita tensor. Equivalently, in this range of ranks, the invariant structures under the connected spatial rotation group are generated solely by metric pairings and timelike factors. The invariant-theoretic argument presented in Appendix \ref{appendix:invariant-proof} makes this point explicit and complements the Schur--Weyl-based proof given in the main text.

As an immediate physical consequence, the rank-two sector of the decomposition implies that, whenever the metric field equations of a modified theory of gravity, whose dynamical fields are imposed to be FLRW compatible, can be rearranged into an Einstein-like form on an FLRW background, the correction term must have the effective perfect-fluid form
\begin{equation}
    Q_{\mu\nu}=\alpha(t)g_{\mu\nu}+\beta(t)u_\mu u_\nu.
\end{equation}
 Thus, the FLRW background equations of such theories are
universally restricted to the perfect-fluid tensorial sector. This statement is
tensorial and background-level: different theories remain dynamically distinct
through the scalar functions $\alpha(t)$ and $\beta(t)$, possible nonminimal
couplings, and the additional field equations of the model. In this sense, the
present work complements the scalar--vector--tensor FLRW universality result of
\cite{SVT-FLRW} and the universality established for purely curvature-driven
theories such as $f(R)$ and higher-order gravity
\cite{PF in f(R), PF in higher order GR}. While the former targets specific field
classes, and the latter deduces the perfect-fluid structure from the curvature
sector on generalized Robertson--Walker backgrounds with divergence-free
conformal curvature, formulated in \cite{PF in extended gravity} through the
notion of a perfect scalar, our approach derives this universality for FLRW as a
rank-two corollary of a top-down, rank-$N$ algebraic decomposition.

We also showed how the decomposition organizes the cosmological field equations and the geodesic deviation equation in modified gravity. Once the correction terms are written in the effective-fluid form, the generalized Friedmann equations and the curvature terms entering the geodesic deviation equation can be expressed in terms of time-dependent effective density and pressure variables. The examples of $f(R,T)$ gravity and Einstein--Aether gravity illustrate how model-dependent information enters through the specific scalar functions determined by the action and by the remaining field equations.

\par As a final remark, throughout our discussion we have assumed the FLRW compatibility of the dynamical fields as an independent postulate, which may appear to restrict the universality of the decomposition theorem. However, in a wide class of gravitational theories, these fields are not merely assumed to be symmetric, but instead dynamically inherit the symmetries of the background spacetime metric. Since the present work is self-contained, a rigorous analysis detailing the exact boundaries of this inheritance (specifically, identifying the precise constraints on the action in \eqref{generalized action} and the intrinsic symmetries of the fields) is deferred to a forthcoming publication. Furthermore, while the current derivation is strictly formulated within the framework of a Levi-Civita connection, this geometric restriction can be generalized to spacetimes featuring non-metricity and non-vanishing torsion under certain compatibility conditions; this extension will likewise be addressed in future work.

\newpage
\section*{Acknowledgment}
The work of YH was supported by the Science Academy’s Distinguished Young Scientist Awards Program (BAGEP)
2025.
\section*{Declarations}

\subsection*{Data Availability Statement}

Data sharing is not applicable to this article, as no datasets were generated, used, or analyzed during the current study.

\subsection*{Conflict of Interest Statement}

The authors declare that they have no conflict of interest.

\section*{Appendix A: Invariant-Theoretic Proof of the Monomial Decomposition}
\label{appendix:invariant-proof}

In this appendix we give an alternative proof of
Theorem~\ref{theorem 1} based directly on the invariant theory of the
spatial isotropy group. This proof is complementary to the Schur--Weyl-based
argument in the main text and makes explicit why, in the rank range considered
here, the only admissible building blocks are the metric and the FLRW time-like
vector.

Fix a point \(p\) in the FLRW spacetime. The cotangent space decomposes as
\begin{equation}
    T_p^*M = \mathrm{span}\{u_\mu\}\oplus \Sigma_p^* ,
\end{equation}
where \(\Sigma_p^*\) is the cotangent space to the homogeneous and isotropic
spatial hypersurface at \(p\). The isotropy group at \(p\) acts trivially on the
one-dimensional time direction generated by \(u_\mu\), and acts in the standard
way on \(\Sigma_p^*\). Thus, an FLRW-compatible covariant tensor
\(T_{\mu_1\cdots\mu_N}\), evaluated at \(p\), defines an invariant element of
\((T_p^*M)^{\otimes N}\) under the spatial isotropy group.

Using the above splitting, each component of \(T_{\mu_1\cdots\mu_N}\) can be
decomposed according to the number \(q\) of spatial slots and \(N-q\) time-like
slots. The time-like slots are represented by factors of \(u_\mu\). Therefore,
the only nontrivial part of the classification is the classification of
invariant covariant tensors on \(\Sigma_p^*\) of rank \(q\).

For the full orthogonal group \(O(D)\), the first fundamental theorem of
invariant theory for the orthogonal group states that invariant multilinear
forms are generated by pairwise contractions with the invariant bilinear form;
equivalently, invariant covariant tensors are spanned by products of the
spatial metric \(h_{ab}\)
\cite{Weyl_Classical_Groups,Goodman_Wallach,Procesi_Lie_Groups}. Hence there
are no invariant tensors of odd spatial rank, while for \(q=2r\) the invariant
tensors are linear combinations of all pairings
\begin{equation}
    h_{a_{i_1}a_{j_1}}\cdots h_{a_{i_r}a_{j_r}},
\end{equation}
where the spatial indices are paired in all distinct possible ways. After
restoring the time-like slots, the spacetime tensor is therefore a linear
combination of monomials made from \(h_{\mu\nu}\) and \(u_\mu\).

Since
\begin{equation}
    h_{\mu\nu}=g_{\mu\nu}+u_\mu u_\nu ,
\end{equation}
any monomial built from \(h_{\mu\nu}\) and \(u_\mu\) can be rewritten as a
linear combination of monomials built from \(g_{\mu\nu}\) and \(u_\mu\).
Therefore,
\begin{equation}
    T_{\mu_1\cdots\mu_N}
    =
    \sum_{r=0}^{\lfloor N/2\rfloor}
    \sum_{\sigma\in \mathcal{S}_{N,r}}
    \alpha_{\sigma r}(t)\,
    M^{(r,\sigma)}_{\mu_1\cdots\mu_N},
\end{equation}
where \(\mathcal{S}_{N,r}\) denotes a set of representatives generating the
distinct \(g-u\) monomials of rank \((N,r)\). This is precisely the monomial
decomposition stated in Theorem~\ref{theorem 1}.

For the connected rotation group \(SO(D)\), one must also consider the spatial
Levi--Civita tensor \(\epsilon_{a_1\cdots a_D}\), which is invariant under
orientation-preserving rotations. This is precisely why the hypothesis \(D>N\)
is imposed in Theorem~ \ref{theorem 1}: for ranks below \(D\), no
Levi--Civita-type invariant can occur. Under this hypothesis, the \(SO(D)\) and
\(O(D)\) classifications coincide for the tensor ranks considered in the
theorem. Finally, the coefficient functions depend only on the cosmological
time because spatial homogeneity eliminates spatial dependence, while isotropy
fixes the allowed spatial tensorial structures.

\end{document}